\documentclass[10pt,twocolumn,twoside]{IEEEtran}

\usepackage{cite}
\usepackage{amsmath,amssymb,amsfonts,amsthm}
\usepackage{algorithmic}
\usepackage{graphicx}
\usepackage{textcomp}
\usepackage{xcolor}
\usepackage{esvect}
\usepackage{url}
\usepackage{mathtools}
\usepackage{physics}
\usepackage{stmaryrd}
\usepackage{subcaption}
\usepackage{enumitem}

\usepackage{tikz}
\usetikzlibrary{3d, shapes, arrows.meta, positioning}

\def\BibTeX{{\rm B\kern-.05em{\sc i\kern-.025em b}\kern-.08em
    T\kern-.1667em\lower.7ex\hbox{E}\kern-.125emX}}

\newtheorem{definition}{Definition}
\newtheorem{theorem}{Theorem}
\newtheorem{proposition}{Proposition}
\newtheorem{corollary}{Corollary}
\newtheorem{lemma}{Lemma}
\newtheorem{example}{Example}

\ifodd 0
\newcommand{\rev}[1]{{\color{blue}#1}} %revise of the text
\newcommand{\revr}[1]{{\color{blue}#1}} %revise of the text
\newcommand{\com}[1]{{\color{red}\textbf{Parinaz's Comment}: #1}}%comment of the text
\newcommand{\comr}[1]{{\color{orange}\textbf{Raman's Comment}: #1}}%comment of the text
\newcommand{\resp}[1]{{\color{cyan}\textbf{Response}: #1}} %responses to comments of the text
\else
\newcommand{\rev}[1]{#1}
\newcommand{\revr}[1]{#1}
\newcommand{\com}[1]{}
\newcommand{\comr}[1]{}
\newcommand{\resp}[1]{}
\fi

\begin{document}
% \twocolumn[ % Method A for two-column formatting
  % \begin{@twocolumnfalse} % Method A for two-column formatting

\title{United We Fall: On the Nash Equilibria of Multiplex and Multilayer Network Games}

\author{Raman Ebrahimi and Parinaz Naghizadeh
% https://orcid.org/0009-0004-6724-6029 Raman
\thanks{This work is supported by the NSF under award CCF-2144283.}
\thanks{R. Ebrahimi and P. Naghizadeh are with the Electrical and Computer Engineering Department, University of California, San Diego. e-mail: \{raman, parinaz\}@ucsd.edu}
}
\maketitle

%%%%%%%%%%%%%%%  Main text   %%%%%%%%%%%%%%%
% \linenumbers
\begin{abstract}
    Network games provide a framework to study strategic decision making processes that are governed by structured interdependencies among agents. However, existing models do not account for environments in which agents simultaneously interact over multiple networks, or when agents operate over multiple action dimensions. In this paper, we propose new models of \emph{multiplex} network games to capture the different modalities of interactions among strategic agents, and \emph{multilayer} network games to capture their interactions over multiple action dimensions. We explore how the properties of the constituent networks of a multiplex/multilayer network can undermine or support the existence, uniqueness, and stability of the game's Nash equilibria. Notably, we highlight that both the largest and smallest eigenvalues of the constituent networks (reflecting their connectivity and two-sidedness, respectively) are instrumental in determining the uniqueness of the multiplex/multilayer network game's equilibrium. Together, our findings shed light on the reasons for the fragility of equilibria when agents interact over networks of networks, and point out potential interventions to alleviate them.
\end{abstract}

\section{Introduction}\label{sec:intro}

%\IEEEPARstart{N}{etworks}
Networks provide a powerful framework for understanding and analyzing real-world environments in which agents interact with and influence one another. In particular, when participating agents are rational and self-interested, the interactions among them can be modeled as a \emph{network game} \cite{jackson2015games}. Such networked strategic interactions emerge in the local provision of public goods~\cite{bramoulle2014strategic,allouch2015private,naghizadeh2017provision} (such as cyber-security, R\&D), spread of shocks in financial markets~\cite{acemoglu2015systemic}, and pricing in the presence of social effects and externalities~\cite{candogan2012optimal}.  

Although existing works in this area capture a network of interactions {between strategic agents}, they fail to capture \emph{the various networks} and \emph{the various action dimensions} over which the agents can interact. For instance, individuals are influenced by information received over multiple social networks, as well as face-to-face interactions, when making decisions. Similarly, firms in a market can cooperate and compete with each other along different modes of business (e.g., physical storefronts and online shops) and across product categories. These scenarios call for more complex models of ``networks of networks'' that can capture the different modalities and dimensions of agents' interactions. 

Specifically, a \emph{multiplex network} model can be used to simultaneously account for the multiple networks of interactions among the agents, and a \emph{multilayer network} can help us account for the different action dimensions (see Figure~\ref{fig:network-illustrations} for an illustration, and Section~\ref{sec:illustrative-example} for a detailed example in the context of interdependent security games). While there is an existing literature on using the formalism of multilayer/multiplex network models, it has primarily focused on questions about percolation and spread of dynamical processes on these networks; the study of games %(i.e., strategic interactions among rational agents) 
on this type of networks remains largely unexplored, with game-theoretical modeling and analysis often identified as an open area of research by surveys of the field~\cite{boccaletti2014structure,aleta2019multilayer,kivela2014multilayer,salehi2015spreading}. Motivated by this, in this paper, we extend the existing models of \emph{single-layer} network games, by proposing \emph{multiplex} and \emph{multilayer} network games. 

\begin{figure}[t]
    \centering
    \begin{subfigure}[b]{0.48\columnwidth}
    \begin{tikzpicture}[scale=0.5, transform shape]

% Rectangles for layers
\fill[blue!10, opacity=0.5] (0.25,-0.75) -- (7,-0.75) -- (6.5,0.75) -- (-0.25,0.75) -- cycle;

\fill[blue!10, opacity=0.5] (0.25,-2.75) -- (7,-2.75) -- (6.5,-1.25) -- (-0.25,-1.25) -- cycle;

% Layer 1 Nodes
\node[circle, draw, fill=blue!30, inner sep=2pt] (A1) at (1,0) {A};
\node[circle, draw, fill=green!30, inner sep=2pt] (B1) at (2.5,0.4) {B};
\node[circle, draw, fill=red!30, inner sep=2pt] (C1) at (4,-0.4) {C};
\node[circle, draw, fill=yellow!30, inner sep=2pt] (D1) at (5.5,0.2) {D};

% Layer 2 Nodes
\node[circle, draw, fill=blue!30, inner sep=2pt] (A2) at (1,-2) {A};
\node[circle, draw, fill=green!30, inner sep=2pt] (B2) at (2.5,-1.7) {B};
\node[circle, draw, fill=red!30, inner sep=2pt] (C2) at (4,-2.4) {C};
\node[circle, draw, fill=yellow!30, inner sep=2pt] (D2) at (6,-2) {D};

% Intra-layer connections for Layer 1
\draw[dashed] (A1) -- (B1);
\draw[dashed] (B1) -- (C1);
\draw[dashed] (C1) -- (D1);
\draw[dashed] (A1) -- (C1);
\draw[dashed] (A1) -- (D1);

% Intra-layer connections for Layer 2
\draw[dashed] (A2) -- (B2);
\draw[dashed] (B2) -- (C2);
\draw[dashed] (C2) -- (D2);
\draw[dashed] (A2) -- (C2);
\draw[dashed] (B2) -- (D2);

% Inter-layer connections
\draw (A1) -- (A2);
\draw (B1) -- (B2);
\draw (C1) -- (C2);
\draw (D1) -- (D2);

\end{tikzpicture}
    \caption{A \textbf{multiplex} network, representing different types of interaction. Nodes are connected to themselves in other layers.}
    \label{fig:multiplex_network}
    \end{subfigure}
    \hspace{0.01in}
    \begin{subfigure}[b]{0.48\columnwidth}
        \begin{tikzpicture}[scale=0.5, transform shape]

% Rectangles for layers
\fill[blue!10, opacity=0.5] (0.25,-0.75) -- (7,-0.75) -- (6.5,0.75) -- (-0.25,0.75) -- cycle;

\fill[blue!10, opacity=0.5] (0.25,-2.75) -- (7,-2.75) -- (6.5,-1.25) -- (-0.25,-1.25) -- cycle;

% Layer 1 Nodes
\node[circle, draw, fill=blue!30, inner sep=2pt] (A1) at (1,0) {A};
\node[circle, draw, fill=green!30, inner sep=2pt] (B1) at (2.5,0.4) {B};
\node[circle, draw, fill=red!30, inner sep=2pt] (C1) at (4,-0.4) {C};
\node[circle, draw, fill=yellow!30, inner sep=2pt] (D1) at (5.5,0.2) {D};

% Layer 2 Nodes
\node[circle, draw, fill=blue!30, inner sep=2pt] (A2) at (1,-2) {A};
\node[circle, draw, fill=green!30, inner sep=2pt] (B2) at (2.5,-1.7) {B};
\node[circle, draw, fill=red!30, inner sep=2pt] (C2) at (4,-2.4) {C};
\node[circle, draw, fill=yellow!30, inner sep=2pt] (D2) at (6,-2) {D};

% Intra-layer connections for Layer 1
\draw[dashed] (A1) -- (B1);
\draw[dashed] (B1) -- (C1);
\draw[dashed] (C1) -- (D1);
\draw[dashed] (A1) -- (C1);
\draw[dashed] (A1) -- (D1);

% Intra-layer connections for Layer 2
\draw[dashed] (A2) -- (B2);
\draw[dashed] (B2) -- (C2);
\draw[dashed] (C2) -- (D2);
\draw[dashed] (A2) -- (C2);
\draw[dashed] (B2) -- (D2);

% Inter-layer connections
\draw (B1) -- (A2);
\draw (B1) -- (C2);
\draw (C1) -- (B2);
\draw (C1) -- (D2);

\end{tikzpicture}
    \caption{A \textbf{multilayer} network, representing different action dimensions. Nodes can connect to any other node in other layers.}
    \label{fig:multilayer_network}
    \end{subfigure}
    \caption{Illustration of multiplex and multilayer networks.}
    \vspace{-0.16in}
    \label{fig:network-illustrations}
\end{figure}
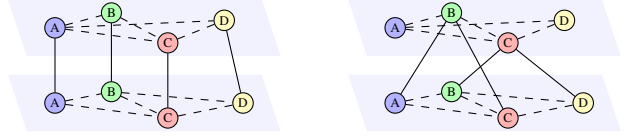

A primary direction of research on single-layer network games has been to analyze how the structural properties of the network of interactions among the agents influences the equilibrium outcomes. In particular, existing works have identified necessary and sufficient conditions for the existence, uniqueness, and/or stability of Nash equilibria (NE) of single-layer network games~\cite{bramoulle2014strategic,naghizadeh2015provision,naghizadeh2017provision,miura2008security,ballester2010interactions,allouch2019constrained,cai2019role,allouch2015private,zhou2016game,naghizadeh2017uniqueness,melo2018variational,parise2019variational,khalili2019public,jin2020games,huanginterdependent}. We similarly identify conditions under which the Nash equilibria of multiplex/multilayer network games are (not) unique, and/or stable, based on the properties of their constituent layers. Through these, we highlight potential reasons for the fragility of the uniqueness and stability of equilibria in these networks. 

\vspace{0.1in}
\subsubsection*{Paper overview and contributions} We consider two subnetworks/layers $\alpha$ and $\beta$, with (weighted and directed) interdependency matrices $G^{(1)}$ and $G^{(2)}$. In multiplex network games, each of these represent one modality or type of interaction (e.g., a person getting information from online vs. in-person connections). As both sources collectively impact agents' actions, the multiplex network game can be viewed as a game with a \emph{sum} interdependency matrix $G^{\parallel}:=\kappa G^{(1)}+ (1-\kappa)G^{(2)}$, where $\kappa\in[0,1]$ captures the effect of each layer on the agent's utility. In multilayer network games, on the other hand, each layer can be viewed as a different action dimension (e.g., a seller enhancing users' storefront vs. online buying experiences). As each agent now has a two-dimensional action space, %and our model is such that 
the multilayer network game can be viewed as a game with a \emph{block} interdependency matrix $G^{\merge}$, with $G^{(1)}$ and $G^{(2)}$ as its diagonal blocks, and off-diagonal blocks $G^{(12)}$ and $G^{(21)}$ capturing potential spillovers between action dimensions. 

Existing approaches to identifying conditions for the uniqueness of Nash equilibrium on single-layer network games \cite{naghizadeh2015provision}, which we build on, require assessing the determinants of the (principal minors) of the game interdependency matrix $G$. The main technical challenge in conducting a similar analysis here is that for either $G^\parallel$ (a sum matrix) or $G^\merge$ (a block matrix), there are no closed form characterizations of their (principal minors') determinants in terms of those of $G^{(1)}$ and $G^{(2)}$. Therefore, it is not straightforward to assess whether having unique NE on the constituent layers translates into uniqueness of NE on a multiplex/multilayer game. 

We however show that it is possible to provide sufficient (resp. necessary) conditions, in terms of the spectral properties of the constituent layers, for an NE to be (resp. not be) guaranteed to be unique on the multiplex/multilayer network game, particularly for the special class of undirected networks (which have symmetric adjacency matrices). We outline such conditions for multiplex networks in Propositions~\ref{prop:pair-failure}-\ref{prop:lambda-min-fails}, and for multilayer networks in Propositions~\ref{prop:twolayer-multilayer-general}-\ref{prop:multilayer-BDD}, each followed by an intuitive interpretation of their implications. We further show that the conditions guaranteeing NE uniqueness also guarantee that the equilibrium will remain stable under perturbations to the agents' utilities (Proposition~\ref{prop:stability}). We illustrate our findings through numerical experiments in Section~\ref{sec:simulations}. All proofs appear at the end of the paper in Appendix~\ref{app:proofs}. Our findings generally scale to networks with more than two layers; we illustrate this in Appendix~\ref{appendix:m-layers}. (We also provide conditions for the \emph{existence} of NE in these games, \rev{and additional numerical experiments}, in the online appendix~\cite{online-appendix}). 

\vspace{0.1in}
\subsubsection*{A key takeaway} We highlight one recurring intuition gained from our analyses: we find that both the lowest ($\lambda_{\min}$) and largest ($\lambda_{\max}$) eigenvalues of the constituent layers' adjacency matrices play a key role in guaranteeing (or undermining) the uniqueness of the equilibria on the multiplex/multilayer network games. Intuitively, $\lambda_{\min}$ is a measure of a network's ``two-sidedness'', with a smaller (more negative) lowest eigenvalue being an indication that agents' actions ``rebound'' more in the network~\cite{bramoulle2014strategic}; $\lambda_{\max}$, on the other hand, is a measure of connectivity. Prior work (e.g.,  \cite{bramoulle2014strategic,naghizadeh2015provision}) had highlighted the role of the lowest eigenvalue in determining whether the NE of a single-layer network game is unique; our work highlights that the largest eigenvalue will also be consequential when agents interact over networks of networks. Intuitively, several of our conditions can be interpreted as assessing whether the connectivity of one layer can (or can not) overcome the rebounds introduced by another layer (or the rebounds across action dimensions in multilayer games). 

These findings can inform network design and interventions. {For instance, our findings suggest that a network designer or policy maker can modify the connectivity in a network layer (e.g., one social network or one industry) by intervening in the connection patterns or intensity of dependencies between its agents, to help induce a unique and/or stable equilibrium in a larger multilayer/multiplex network; this may be advantageous if interventions in all layers are not feasible. Alternatively, a policy-maker could assess the consequences of two previously independent sectors becoming interdependent (two single-layer networks joining to form a multilayer network) by using knowledge about the properties of each of the existing layers.}

\subsection{{Illustrative example}}\label{sec:illustrative-example}

Consider a set of $N$ firms, interacting with each other over a network $\alpha$, with the network representing dependencies between the firms (e.g., shared infrastructures or joint operations). Each firm can make an investment to strengthen its cybersecurity, {which impacts not only the firm but also others in its network.} These security investments can be a strategic \emph{substitute} when a better protected firm $j$ positively impacts other firms $i$ that share operations and assets with firm $j$, by decreasing the risk of business interruption or asset compromise. On the other hand, security investments can be a strategic complement when an increase in firm $j$'s protection makes a similar, but less protected, firm $i$ a more attractive target for attackers. In either case, the edge weights of network $\alpha$ can capture these types of dependencies. These strategic interactions have been commonly studied as \emph{interdependent security games} \cite{Jiang2011,naghizadeh2016opting} on a (single-layer) network. 

\emph{Multiplex network:} Consider a second network $\beta$, with the same set of firms as its nodes, but with the network capturing direct dependencies in firms' operations in another industry sector. For instance, layers $\alpha$ and $\beta$ could capture the communication and energy sectors. Then, the security effort of each firm can impact security spillovers and disruptions in both these sectors. A multiplex network can be used to capture such scenarios, and provides a more holistic understanding of the impacts of each firm's security investments. {Policy makers can leverage such models to understand the reasons for potential instabilities in the equilibria of interconnected networks (e.g., due to one sector being not sufficiently well-connected, or one sector being two-sided with high rebound/oscillations in agents' decisions). They can then design regulations that limit or shape how the two sectors interact and induce a desirable setting where the firms' actions will settle at a stable outcome.}

\emph{Multilayer network:} Alternatively, assume a second network $\beta$ captures the exchange of information related to cybersecurity threats and vulnerabilities between firms. Each firm makes two types of decisions: investments in cybersecurity in layer $\alpha$, and level of information sharing in layer $\beta$. Both these decisions concurrently impact the dependent firms' utilities. A multilayer network structure can be used to capture this type of scenario. By incorporating both direct dependencies in network $\alpha$ and information sharing decisions in network $\beta$, the combined framework offers a more nuanced understanding of the multifaceted interdependencies influencing firms' decisions in allocating their efforts in each dimension. {The methods presented here could enable each firm to assess the impact of its membership in different information sharing agreements on the (existence, uniqueness or stability of equilibrium) decisions of other firms regarding information sharing or security investments. Policy makers could similarly introduce regulations on information sharing to induce stable/unique outcomes.}

\subsection{Related literature}
Our work is at the intersection of two lines of literature: (i) the study of properties of Nash equilibria of single-layer network games, and (ii) game-theoretical modeling and analysis of multilayer and multiplex networks.

\emph{Nash equilibria in single-layer games.} Specifically, our work closely relates to studies that investigate the existence, uniqueness, and stability of Nash equilibria in single-layer networks games with linear~\cite{bramoulle2014strategic, naghizadeh2017provision, miura2008security, ballester2010interactions, allouch2019constrained} and nonlinear~\cite{naghizadeh2017uniqueness, parise2019variational, melo2018variational} best-replies. 
Our proposed models of multiplex and multilayer network games are extensions of single-layer network games with linear best-replies, by allowing agents to have various types of interactions (multiplex) and operate across different action dimensions (multilayer). For these, we identify conditions under which the multiplex/multilayer network will (or will not) 
inherit the guarantees of uniqueness/stability of the Nash equilibrium in its constituent layers.

\emph{Games on multilayer/multiplex networks.} While game-theoretical decision making on multiplex and multilayer networks has also been studied in some prior works~\cite{boccaletti2014structure,shahrivar2017game,Gómez-Gardeñes2012,Battiston_2017,CHEN2021125532,LI2021126110}, the majority of the research has focused on evolutionary games and the emergence of cooperation in public good games. In particular, the emergence and sustainability of cooperation strategies in public good games has been explored in multiplex networks by \cite{Gómez-Gardeñes2012,Battiston_2017} and in multilayer networks by \cite{CHEN2021125532,LI2021126110}, 
assuming a binary action (cooperate/defect) game, and often restricting focus to specific classes of graphs. In contrast, we propose  and analyze more general classes of multiplex and multilayer networks games, allow for arbitrary directed and weighted interactions among agents, and provide insights into the games' Nash equilibria properties.

The recent work of \cite{jin2024structured} is also closely related to ours, as it proposes a ``multi-relational'' network game which is similar to our model of multilayer games, in that it allows for a multi-dimensional action space. The main focus of \cite{jin2024structured} is on identifying ``summary representations'' of the game matrix that can be used to significantly lower the \emph{computation complexity} of ascertaining the uniqueness of NE (\revr{motivated by the fact that verifying NE uniqueness is in general NP-hard \cite{conitzer2003complexity}.}) \revr{While our focus is different—illustrating how the properties of the constituent layers of a multilayer (or multiplex) network can undermine or support the uniqueness of NE) in these games— our findings can also help lower the computational costs of verifying NE uniqueness on large multilayer/multiplex networks by assessing conditions on the eigenvalues of smaller dimensional networks (each layer) instead.} 

This paper extends our earlier work in  \cite{ebrahimi2023multiplex}, which investigated the uniqueness of Nash equilibria in multiplex network games. The current work provides additional conditions for uniqueness of NE in multiplex network games, introduces the new class of multilayer network games and its NE uniqueness conditions, establishes new results on stability and existence of NE in both multiplex and multilayer games, and illustrates our findings through numerical experiments. 

\section{Model}\label{sec:model}

\subsection{Single-layer network games}\label{sec:single-layer-model}

Consider a set of $N$ agents interacting with each other over a {single-layer} network $\alpha$. % (e.g., a cyber-physical system, social network, or industry). 
This network is specified by a graph $\mathcal{G}_\alpha \coloneqq \langle\mathcal{V}, {G^{(1)}}\rangle$, where the $N$ agents constitute the set of vertices $\mathcal{V}$, and ${G^{(1)}}$ is the weighted and directed \emph{adjacency} or \emph{interdependency matrix} over network $\alpha$. 

Each agent $i$ selects an \emph{effort} level $x\in\mathbb{R}_{\geq 0}$. %; this could represent the amount of investment in a public good such as cyber security or R\&D. 
The agent's utility is determined by its own effort, as well as the effort of its neighboring agents in the network. 
Specifically, {an edge in ${G^{(1)}}$, with weight $g^{(1)}_{ij}$}, indicates that agent $i$ is affected by agent $j$'s effort. If $g^{(1)}_{ij} > 0$ (respectively, $< 0$), we say agent $j$'s effort is a \emph{substitute} (respectively, \emph{complement}) to agent $i$'s effort. In our setting, a strategic substitute (resp. complement) means that effort by agent $j$ provides positive (resp. negative) externality to agent $i$, in that an increase in effort by agent $j$ allows agent $i$ to decrease its own effort (resp. requires agent $i$ to increase its effort) and still receive the same overall payoff. If there is no influence on agent $i$ from agent $j$, then $g_{ij}=0$. \revr{Further, we assume there are no self-loops; i.e., $g_{ii}=0, \forall i$.}
%For instance, security investments can be a strategic substitute when a better protected firm $j$ positively impacts other firms $i$ that share operations and assets with firm $j$, by decreasing the risk of business interruption or asset compromise. On the other hand, security investments can be a strategic complement when an increase in firm $j$'s protection makes a similar, but less protected firm $i$ a more attractive target for attackers. 

Let {$\mathbf{x}\in \mathbb{R}^{N\times 1}$} denote the vector of all agents' efforts. Agent $i$'s utility in network $\alpha$ is given by:
\begin{align}
    u_{i} (\mathbf{x}; {G^{(1)}}) = b_i(x_i + \sum_{{j=1}}^{{N}}g^{(1)}_{ij}x_j) - c_ix_i~,
    \label{eq:utility-alpha}
\end{align}
where {$b_i(\cdot):\mathbb{R}\rightarrow \mathbb{R}$ is a twice-differentiable, strictly increasing, and strictly concave \emph{benefit function}}, which has as its argument the \emph{aggregate effort} experienced by the agent, and {$c_i>0$ is the \emph{unit cost} of effort for agent $i$.} 

The (single-layer) network game specified by the set of $N$ agents, their efforts $\mathbf{x}$, and their utility functions $\{u_{i}(\mathbf{x}; {G^{(1)}})\}$ has been studied extensively in prior works (e.g.,~\cite{naghizadeh2017provision,bramoulle2014strategic,miura2008security,Jiang2011}). In particular, these games are known as games of \emph{linear best-replies}, as the Nash equilibrium $\mathbf{x}^*$ is determined by a set of linear best-response equations of the form:
\begin{align}
    x_i^* = \max\{0, q_i-\sum_{{j=1}}^{{N}} g^{(1)}_{ij}x^*_j\}~,
    \label{eq:linear-br}
\end{align}
where $q_i$ satisfies $b_i'(q_i)=c_i$, {with $b_i'(q_i)$ denoting the first derivative of $b_i(\cdot)$ evaluated at $q_i$. These \emph{target aggregate effort levels} $q_i$ exist and are unique given the assumptions on the functions $b_i(\cdot)$ and the unit costs $c_i$.} Intuitively, an agent $i$ wants to receive an aggregate effort level $q_i$ at equilibrium; this is the effort level at which the agent's marginal benefit and marginal cost of effort are equalized. The best-response \eqref{eq:linear-br} states that the agent exerts effort $x_i^*$ to reach an aggregate effort level $q_i$, given the spillover $\sum_{{j=1}}^{{N}} g^{(1)}_{ij}x^*_j$ received from its neighboring agents' effort at equilibrium, or exerts no effort if the spillovers already provide aggregate effort $q_i$ or higher. 

{We next propose two extensions of these existing models: multiplex network games, and multilayer network games.} 
%In both cases, we assume agents can further interact over a second network defined by a graph $\mathcal{G}_\beta \coloneqq \langle\mathcal{V}, G^{(2)}\rangle$, with the same set of vertices as network $\alpha$, but its own interdependency matrix ${G^{(2)}}$. 
We present \emph{two-layer} multiplex and multilayer network games; extensions to games involving more layers is straightforward. %\footnote{We illustrate an extension to the specific subclass of chain/cycle multilayer networks in Appendix~\ref{appendix:m-layers}.} 

\subsection{Multiplex network games}\label{sec:multiplex-model}

Consider a second network defined by a graph $\mathcal{G}_\beta \coloneqq \langle\mathcal{V}, G^{(2)}\rangle$, with the same set of vertices as network $\alpha$, but its own interdependency matrix ${G^{(2)}}$.  
The two-layer multiplex network $\mathcal{G}^{\parallel} \coloneqq {\langle\mathcal{V}}, \{G^{(1)}, G^{(2)}\}\rangle$ is the environment in which interactions between the $N$ agents occur over both networks $\alpha$ and $\beta$ simultaneously, but each governed by a different interdependency matrix. 
Let the utility of agent $i$ in the multiplex network be given by:
\begin{align}
    &u_{i} (\mathbf{x}; G^{(1)}, G^{(2)}, \kappa) =\notag\\
    &\quad b_i(x_i + \kappa\sum_{{j=1}}^{{N}}g^{(1)}_{ij}x_j + (1-\kappa)\sum_{{j=1}}^{{N}} g^{(2)}_{ij}x_j) - c_ix_i,
    \label{eq:utility-multiplex}
\end{align}
where $\kappa\in[0,1]$ captures the effect of each layer on the utility, with higher $\kappa$'s indicating higher effects from network $\alpha$.\footnote{\revr{$\kappa$ provides a degree of freedom to the model designer, allowing them to adjust the ``importance'' of each layer in determining agents' actions/utilities. For instance, when the layers represent social networking platforms, $\kappa$ could be assessed based on the amount of time agents spend interacting with one another on each social network.}}

The resulting \emph{multiplex network game} is again a game of linear best-replies, where agent $i$ aims to choose $x_i^*$ to reach the same aggregate level of effort $q_i$, but this time while being exposed to spillovers $\kappa\sum_{{j=1}}^{{N}}g^{(1)}_{ij}x_j^* + (1-\kappa)\sum_{{j=1}}^{{N}} g^{(2)}_{ij}x_j^*$ from the multiplex network. Then, the multiplex network game can be viewed as a network game played over the interdependency matrix $G^{\parallel} \coloneqq \kappa G^{(1)} + (1-\kappa)G^{(2)}$.

\subsection{Multilayer network games}\label{sec:multilayer-model}
Consider again the second network $\mathcal{G}_\beta \coloneqq \langle\mathcal{V}, G^{(2)}\rangle$, with the same set of nodes as network $\alpha$.\footnote{We assume the layers have the same set of nodes to simplify notation. Our results will continue to apply when the set of nodes are different.} Assume now that agents take a different action in each layer. Let $x_i^{(l)}$ denote agent $i$'s action in layer $l\in\{1, 2\}$, and $\mathbf{x}^{(l)}$ denote the vector of all agents' efforts in layer $l$. Additionally, assume there are dependencies between agents' actions in different layers, represented by inter-layer dependency matrices $G^{(lk)}$. An edge $g^{(lk)}_{ij}$ captures how agent $i$'s utility in layer $l$ is impacted by agent $j$'s action in layer $k$.

\revr{Agent $i$'s aggregate utility from the two layers is given by $u_i=\kappa u_i^{(1)}+(1-\kappa)u_i^{(2)}$, where $\kappa\in[0,1]$ captures the effect of each layer on the agent's total utility, and agent $i$'s utility from layer $l$ is given by:}
\begin{align}\label{eq:utility-multilayer-layer}
    &u_i^{(l)}({\mathbf{x}}^{(l)}, {\mathbf{x}}^{(k)}; G^{(l)}, G^{(lk)})=&\notag\\
    &\qquad b_i^{(l)}(\,x_i^{(l)}+\sum_{{j=1}}^{{N}} g^{(l)}_{ij}x^{(l)}_j+\sum_{{j=1}}^{{N}} g^{(lk)}_{ij}x^{(k)}_j\,)-c^{(l)}_ix_i^{(l)}, 
\end{align}
where $k$ is the index of the other layer, and $b_i^{(l)}(\cdot)$ and $c_i^{(l)}$ are the benefit function and unit cost of action dimension $l$, respectively. The argument of the benefit function is the aggregate effort in action dimension $l$ experienced by the agent, and it comes from three sources: the agent's own effort in action dimension $l$, the intra-layer spillovers from action dimension $l$ of neighbors in layer $l$, and the inter-layer spillovers from action dimension $k$ of neighbors in layer $k$.

Accordingly, agent $i$'s best-responses in each action dimension are given by
\begin{align}
    {x}_i^{*(1)} = \max\{0, {q}^{(1)}_i-\big(\begin{pmatrix}
        G^{(1)} & G^{(12)}
    \end{pmatrix} \begin{bmatrix}
        \mathbf{x}^{*(1)}\\
        \mathbf{x}^{*(2)}
    \end{bmatrix}\big)_i\}~,\notag\\
    {x}_i^{*(2)} = \max\{0, {q}^{(2)}_i-\big(\begin{pmatrix}
        G^{(21)} & G^{(2)}
    \end{pmatrix} \begin{bmatrix}
        \mathbf{x}^{*(1)}\\
        \mathbf{x}^{*(2)}
    \end{bmatrix}\big)_i\}~,
    \label{eq:multilayer-br-layer}
\end{align}
with $q^{(l)}_i$ satisfying ${b^{(l)}_i}'(q^{(l)}_i)=c_i^{(l)}$. Or compactly: 
\begin{align}
    \mathbf{x}_i^* = \max\{\mathbf{0}, \mathbf{q}_i-(G^\merge\mathbf{x}^*)_{[i,i+N]}\}~,
    \label{eq:multilayer-br}
\end{align}
where the max operator is element-wise, $(G^\merge\mathbf{x}^*)_{[i,i+N]}$ is the $2\times 1$ vector consisting of the $i^\textsuperscript{th}$ and $(i+N)^\textsuperscript{th}$ entries of the $2N\times 1$ vector $G^\merge\mathbf{x}^*$, and 
\begin{align}
    G^{\merge} \coloneqq \begin{pmatrix}\label{blockmat}
    G^{(1)} & G^{(12)}\\
    G^{(21)} & G^{(2)}
    \end{pmatrix}
\end{align}
is a $2N\times2N$ \emph{supra-adjacency} matrix. Then, the multilayer network game can be viewed as a network game with a two-dimensional action space played over the matrix $G^{\merge}$. 

\subsection{Preliminaries: Uniqueness of Nash equilibria of single-layer network games}\label{sec:single-layer-prelim}
%\subsection{Uniqueness of NE of single-layer games}\label{sec:single-layer-prelim}

We first review existing conditions on the game adjacency matrix $G$ under which the equilibria of single-layer network games are (guaranteed to be) unique. We will later evaluate these conditions on $G^\parallel$ (for multiplex networks, Section \ref{sec:multiplex}) and $G^\merge$ (for multilayer networks, Section \ref{sec:multilayer}), identifying when they do (not) hold in terms of the properties of the layers $G^{(1)}$ and $G^{(2)}$, and inter-layer interactions $\{G^{(12)},G^{(21)}\}$. 

We build on the findings of \cite{naghizadeh2017provision}, which explored the connection between finding the Nash equilibrium of games of linear best-replies and linear complementarity problems (LCPs), to identify conditions for the existence and uniqueness of the NE of single-layer network games. Formally, an LCP($M$, $\mathbf{b}$) seeks to find two $N\times 1$ vectors, $\mathbf{w}$ and $\mathbf{z}$, satisfying:
\begin{align}\label{LCPdef}
    &\mathbf{w}-M\mathbf{z} = \mathbf{b}\notag\\
    &\mathbf{w}\ge 0 \,,\;\mathbf{z}\ge 0 \,,\;\text{and}\; w_iz_i=0, \forall i=1, \ldots, N 
\end{align}
where $M$ is an $N\times N$ matrix, $\mathbf{b}$ is an $N\times 1$ vector, \revr{and  the inequalities indicate element-wise non-negativeness of the vectors}. Comparing this with \eqref{eq:linear-br}, we observe that finding the NE of a single-layer network game is equivalent to solving the LCP$(I+G, -\mathbf{q})$. This equivalence allows us to leverage existing characterizations of the properties of LCP solutions to assess a network game's NE properties. %{It is well known, shown in \cite{Chung1989}, that LCP is NP-complete. In the following sections we introduce conditions in which by connecting the networks, the matrix of the resulting network game remains a P-matrix, i.e. needless to check the P-property.}

To present conditions for NE uniqueness, we begin with the following definition:
%\vspace{0.1in}
\begin{definition}
A square matrix $M$ is a \emph{P-matrix} (denoted $M\in \mathcal{P}$) if the determinants of all its principal minors (i.e., the square sub-matrices obtained from $M$ by removing a set of rows and their corresponding columns) are strictly positive.
\end{definition}
%\vspace{0.1in}

The class of P-matrices includes positive definite (PD) matrices as a special case;\footnote{A common convention adopted in some of the literature is to define positive definiteness for symmetric (or Hermitian) matrices, owing to their roots in quadratic forms. However, we adopt the more general definition here: A square matrix $M$ (whether symmetric or not) is positive definite if $x^TMx>0$ for all $x\neq 0$.} in particular, every PD matrix (whether symmetric or not) is a P-matrix, but there are (asymmetric) P-matrices that are not PD \cite{murty1988linear}. We also note that for symmetric matrices, the two notions are equivalent, i.e., a symmetric matrix is a P-matrix if and only if it is PD.

The following theorem provides the necessary and sufficient condition for the Nash equilibrium of a single-layer network game (whether symmetric or not) to be unique. 

\vspace{0.05in}
\begin{theorem}\label{thm:single-uniqueness}\cite[Theorem 1]{naghizadeh2017provision}
    The single-layer network game with an interdependency matrix $G$ has a unique Nash equilibrium if and only if $I+G$ is a P-matrix.
\end{theorem}
\vspace{0.05in}

The following corollary is the special case of Theorem~\ref{thm:single-uniqueness} for {undirected networks (with symmetric adjacency matrices)}.

\vspace{0.05in}
\begin{corollary}\label{cor:single-uniqueness-symmetric}
\cite{bramoulle2014strategic},\cite[Corollary 1]{naghizadeh2017provision} The single-layer network game with a symmetric interdependency matrix $G$ has a unique Nash equilibrium if and only if $|\lambda_{\min}(G)|<1$. 
\end{corollary}
%\vspace{0.1in}

\section{Multiplex Network Games}\label{sec:multiplex}
In this section, we analyze the uniqueness of Nash equilibria of multiplex network games, in terms of the properties of $G^{(1)}$ and $G^{(2)}$ of their constituent layers. 

To leverage the result of Theorem \ref{thm:single-uniqueness}, we need to check when $I+G^\parallel$, the weighted sum of $I+G^{(1)}$ and $I+G^{(2)}$, is a P-matrix. This will require us to check that the determinants of all principal minors of $I+G^\parallel$ are positive. However, the determinant of the sum of two square matrices $G^{(1)}$ and $G^{(2)}$ is in general not expressible in terms of the determinants of the two matrices.\footnote{The \emph{Marcus–de Oliveira determinantal conjecture}, which conjectures that the determinant of the sum of two matrices is in a convex hull determined by the eigenvalues of the two matrices, remains as one of the open problems in matrix theory, with the conjecture shown to hold for some special classes including Hermitian matrices~\cite{fiedler1971bounds}.} This means that knowledge of the P-matrix property of $I+G^{(1)}$ and/or $I+G^{(2)}$ does not necessarily help establish the P-matrix property for $I+G^\parallel$. 
In fact, the following example shows that the sum of two P-matrices is \emph{not} always a P-matrix. 

\vspace{0.05in}
\begin{example}\label{ex:sum-of-p-matrices}
Consider a two-agent multiplex network game with a benefit function $b_i(x)=1-\exp(-x)$ and unit costs of effort $c_1=\frac{1}{e}$ and $c_2=\frac{1}{\sqrt{e}}$. Let the layers have (asymmetric) interdependency matrices $G^{(1)}=\begin{psmallmatrix}0 & 4 \\ \frac{1}{5} & 0 \end{psmallmatrix}$ and $G^{(2)}=\begin{psmallmatrix}0 & 0 \\ 1 & 0 \end{psmallmatrix}$. Then, $I+G^{(1)}$ and $I+G^{(2)}$ are P-matrices, and each layer by itself has a unique NE. Let $\kappa=0.5$. Then, $I+G^\parallel=\begin{psmallmatrix} 1 & 2 \\ \frac12(\frac{1}{5}+1) & 1 \end{psmallmatrix}$. This is not a P-matrix, and the multiplex has two Nash equilibria: $\mathbf{x}^* = \begin{psmallmatrix}
     0 \\ 0.5
 \end{psmallmatrix}$ or $\mathbf{x}^* = \begin{psmallmatrix}
     1 \\ 0
 \end{psmallmatrix}$. 
\end{example}

\subsection{When is $I+G^\parallel$ not a P-matrix?}\label{sec:negative-results-p-matrix}

We now generalize the intuition from Example~\ref{ex:sum-of-p-matrices} to identify conditions under which $I+G^\parallel$ is not a P-matrix. In particular, note that $I+G^\parallel$ is a P-matrix if and only if the determinant of \emph{all} its principal minors are positive. The following proposition identifies conditions under which at least one of the principal minors has a non-positive determinant. 

\vspace{0.05in}
\begin{proposition}\label{prop:pair-failure}
    Let $M_{ij}^l$ be the $2\times 2$ minor obtained by removing all rows and columns except $i$ and $j$ from $G^{\parallel}$. If there exists a pair of agents $i$ and $j$ such that 
\begin{align*}
    &\frac{g^{(1)}_{ij}}{g^{(2)}_{ij}}(1-\det(M^\alpha_{ij}))+\frac{g^{(2)}_{ij}}{g^{(1)}_{ij}}(1-\det(M^\beta_{ij})) \notag\\
    & \qquad \geq 2 + \frac{\kappa}{1-\kappa}\det(M^\alpha_{ij}) + \frac{1-\kappa}{\kappa}\det(M^\beta_{ij})~,
    %\label{eq:pair-failure-condition}
\end{align*}
then the multiplex network is {not guaranteed} to have a unique Nash equilibrium.\footnote{That is, there are benefit functions and unit costs for which the multiplex network either does not have a Nash equilibrium, or has multiple equilibria.} 
\end{proposition}
\vspace{0.05in}

\emph{Intuitive interpretation.} It is worthwhile to again note that even if both layers satisfy the P-matrix condition, so that the sub-determinants $\det(M^l_{ij})$ are positive, the condition of Proposition~\ref{prop:pair-failure} can still hold at sufficiently large $g^{(1)}_{ij}$ (the determinant term can be kept constant by adjusting $g^{(1)}_{ji}$ accordingly), so that the multiplex will not satisfy the P-matrix condition. Intuitively, this means that the NE may not be unique if there are sufficiently large cyclical dependencies between a pair of agents across the two layers. 
Proposition~\ref{prop:pair-failure} can also be extended to state conditions on higher order principal minors (e.g., highlighting that non-uniqueness can be caused by cyclical interactions among a set of three agents). 

\subsection{When is $I+G^\parallel$ a P-matrix?}\label{sec:positive-results-p-matrix}

While, as shown above, the sum of two P-matrices is in general not a P-matrix, there are specific {subclasses of P-matrices} which are closed under summation, as shown next.  

\vspace{0.05in}
\begin{proposition}\label{prop:sum-of-p-matrices}
For any of the following cases, the multiplex network game will have a unique Nash equilibrium: 
\begin{enumerate}
    \item $G^{(1)}$ and $G^{(2)}$ are \emph{symmetric} {P-matrices}.  
    \item $I+G^{(1)}$ and $I+G^{(2)}$ are \emph{strictly row diagonally dominant}, i.e.,  $\sum_{j\neq i} |g^{(l)}_{ij}| < 1, \forall i, l\in\{1, 2\}$. 
    \item $I+G^{(1)}$ and $I+G^{(2)}$ are \emph{B-matrices}, i.e.,  $1 + \sum_{j} g^{(l)}_{ij} > 0$ and $\frac{1}{N}(1+\sum_{j\neq k} g^{(l)}_{ij}) > g^{(l)}_{ik}$, $\forall i, \forall k\neq i, l\in\{1, 2\}$. 
\end{enumerate}
\end{proposition}
\vspace{0.05in}

\emph{Intuitive interpretation.}  We delve deeper into the case of symmetric matrices in the next subsection. The remaining two cases
%, row diagonally dominant and B-matrices, 
set limits on the influence of agents on each other; Proposition~\ref{prop:sum-of-p-matrices} notes that these limits will carry over when two networks connect. In particular, a row diagonally dominant matrix limits the cumulative maximum influence of neighboring agents on an agent $i$'s utility, relative to the agent's self-influence (here, normalized to 1). If the externalities received from other agents are limited in both layers, they will also be limited when layers are interconnected. B-matrices on the other hand require that the row averages dominate any off-diagonal entries, meaning that no one neighbor's externality on agent $i$ is higher than the average of all the other influences the agent experiences (both self-influence and the externality from the remaining neighbors). Again, if this is true in both layers, it will remain true when the two layers are interconnected as well, guaranteeing NE uniqueness. 

\subsection{Special case: symmetric matrices} We now turn to the special case of {undirected networks}. We begin by noting that the sum of two positive definite matrices is a positive definite matrix. That is, if we know that two {undirected} layers already have structures that are conducive to unique NE, so will the {undirected} multiplex network emerging from joining them. In light of this, we focus on situations where an {undirected} first layer supports a unique NE, yet the second layer does not. We then identify conditions under which the multiplex is (Proposition \ref{prop:multiplex-pert-pos}) and is not (Proposition \ref{prop:lambda-min-fails}) guaranteed to have a unique NE. 

We begin with a positive result: a multiplex may retain NE uniqueness under the following condition. 

\vspace{0.05in}
\begin{proposition}\label{prop:multiplex-pert-pos}
    In a multiplex network game where the first layer is such that $I+G^{(1)}$ is a symmetric positive definite matrix, if the following inequality holds, then the multiplex will have a unique Nash equilibrium:
    \begin{align*}
    %\label{eq:multiplex-pos-res}
    \lambda_{\max}(G^{(2)})<\frac{2\kappa-1}{1-\kappa}+\frac{\kappa}{1-\kappa}\lambda_{\min}(G^{(1)})~.
    \end{align*}
\end{proposition}
\vspace{0.05in}

% \vspace{0.1in}
\emph{Intuitive interpretation.} Proposition~\ref{prop:multiplex-pert-pos} is most useful when the multiplex closely resembles the first layer, as the second layer is comparatively ``weaker''. 
%{The reason behind this is that this result is based on perturbation and if the impact of the constructed multiplex is not in a certain neighborhood of the first layer, then Proposition~\ref{prop:multiplex-pert-pos} will not be able to guarantee uniqueness.} 
Proposition~\ref{prop:multiplex-pert-pos} states a condition for quantifying this relative weakness. This metric evaluates the strength (connectivity) of the second layer relative to the bipartiteness of the first layer. Specifically, if the first layer is highly bipartite (with the smallest eigenvalue approaching $-1$), then the second layer should introduce minimal changes to the network by being relatively sparse (have a small largest eigenvalue). Conversely, if the first layer exhibits less bipartiteness, it can accommodate a stronger connectivity in the second layer. This is because the potential fluctuations caused by the second layer can be mitigated by the structural resilience of the first layer.

Now we present a negative result, where the second layer can undermine the uniqueness of the NE of the multiplex.  

\vspace{0.05in}
\begin{proposition}\label{prop:lambda-min-fails}
    In a multiplex network game with {undirected} layers, if
    \begin{align*}
    &|\lambda_{\min}(G^{(2)})|\geq \frac{1}{1-\kappa} \left(1+  \kappa\lambda_{\max}(G^{(1)})\right)~,
\end{align*}
the game is not guaranteed to have a unique Nash equilibrium. 
\end{proposition}
\vspace{0.05in}

% \vspace{0.1in}
\emph{Intuitive interpretation.} Given that the lowest eigenvalue can be interpreted as a measure of a network's ``two-sidedness'' (with a smaller (more negative) lowest eigenvalue being an indication that agents' actions rebound more in a network), it is expected that a second layer with a large $|\lambda_{\min}(G^{(2)})|$ will introduce similar effects in the multiplex network. The condition in Proposition~\ref{prop:lambda-min-fails} shows that this is indeed the case: when the second layer network is significantly two-sided, it can undermine the uniqueness of the equilibrium in the multiplex network. Also, as expected, for large $\kappa$ (when the second layer is less important in determining agents' payoffs), $\lambda_{\min}(G^{(2)})$ will have less influence on the NE uniqueness. 

More interestingly, the severity of rebound effects due to the second layer network (its $\lambda_{\min}$) are compared against the extent of \emph{connectivity} of the first layer network (its largest eigenvalue $\lambda_{\max}$). In words, Proposition~\ref{prop:lambda-min-fails} states that if the connectivity of the first layer (as characterized by its largest eigenvalue) is not high enough to mute the ups and downs introduced by the second layer (as characterized by its smallest eigenvalue), then the multiplex will have either no equilibrium or multiple equilibria for some game instances. 

\section{Multilayer Network Games}\label{sec:multilayer}
We now identify conditions for guaranteed uniqueness (or lack thereof) of Nash equilibria on multilayer network games, in terms of the properties of their constituent layers $G^{(1)}$ and $G^{(2)}$ and the inter-layer interactions $\{G^{(12)}, G^{(21)}\}$. 

\subsection{When is $I+G^\merge$ not a P-matrix?}\label{sec:negative-results-p-matrix-multilayer}
First, we show that in a general network, with both inter-layer and intra-layer edges being directed and weighted, a Nash equilibrium may not be unique even if each layer is guaranteed to have a unique Nash equilibria by itself. 

\vspace{0.05in}
\begin{proposition}\label{prop:twolayer-multilayer-general}
    A multilayer network game where
    \begin{enumerate}
        \item either of the layers is not guaranteed to have a unique Nash equilibrium (i.e., $\exists l, \text{s.t. }I+G^{(l)}\notin \mathcal{P}$); or
        \item each layer has a unique NE (i.e., $I+G^{(l)}\in \mathcal{P}, \forall l$), but
        \begin{align*}
        %\label{eq:blockdeterminant-condition}
        &\det(I+G^{(1)}-G^{(12)}(I+G^{(2)})^{-1}G^{(21)})\leq 0, \text{  or, }\notag\\
        &\det(I+G^{(2)}-G^{(21)}(I+G^{(1)})^{-1}G^{(12)})\leq 0,
    \end{align*}
        %The determinant of the corresponding Schur complement (for each layer) is non-positive.
    \end{enumerate}
    is not guaranteed to have a unique Nash equilibrium. 
\end{proposition}
\vspace{0.05in}

% \vspace{0.1in}
\emph{Intuitive interpretation.} The first case of the proposition is straightforward: it notes that if the game is not guaranteed to have a unique equilibrium in one action dimension, then it is not guaranteed to have a unique equilibrium in the two dimensional action space. The second case is perhaps more interesting; it argues that despite having a unique equilibrium in each action dimension separately, when the spillovers between action dimensions (due to the inter-layer connections) are taken into account, the game may no longer have a unique NE. The conditions in the second case of Proposition~\ref{prop:twolayer-multilayer-general} capture the extent of spillovers needed for this to happen. 

{To further illustrate, consider the following. From \cite{zhan2005determinantal}, for two non-singular matrices $A$ and $B$ such that $C=B^{-1}A$ is positive semi-definite, $\det(A+B)\ge\det(A)+\det(B)$. Let $A=I+G^{(1)}-G^{(12)}(I+G^{(2)})^{-1}G^{(21)}$ and $B=-(I+G^{(1)})$, {and assume these satisfy the conditions above.}
Then, if 
$\det(I+G^{(1)})\det(I+G^{(2)})\leq \det(G^{(12)})\det(G^{(21)})$, 
the condition in the second case of Proposition \ref{prop:twolayer-multilayer-general} would hold, and the multilayer game would not be guaranteed to have a unique NE. 
This can happen, e.g., if the edge weights in one or both of the inter-layer interactions are sufficiently large and of the same sign both ways. Intuitively, this can be interpreted as large rebound effects of action dimensions across the two layers, which can undermine equilibrium uniqueness.}  

\subsection{When is $I+G^\merge$ a P-matrix?}\label{sec:positive-results-p-matrix-multilayer}
In light of the negative result in the general case, we next explore two special cases when the multilayer network can be guaranteed to have a unique NE: one-way inter-layer connections, and (sufficiently bounded) games of complements. 

\subsubsection{One-way inter-layer connection} Consider a multilayer network in which links are directed only from one layer to the other, so that either $G^{(12)}=\textbf{0}$ or $G^{(21)}=\textbf{0}$.\footnote{For instance, in our illustrative example of interdependent security games (Section~\ref{sec:illustrative-example}), this could happen if information sharing decisions impact security investment decisions, but not vice versa (due to, e.g., information sharing already being mandatory).} For this special case, the supra-adjacency matrix $G^\merge$ is a \emph{block triangular matrix}, which we leverage to establish the following. 
\vspace{0.05in}
\begin{proposition}\label{prop:multilayer-oneway}
    A one-way multilayer network is guaranteed to have a unique Nash equilibrium if and only if the constituent layers are guaranteed to have unique NE (i.e., $I+G^{(l)}\in \mathcal{P}$).
\end{proposition}
\vspace{0.05in}
% \vspace{0.1in}
\emph{Intuitive interpretation.} We note that while the above result is independent of the strength of the one-way inter-layer links, it is highly sensitive to the inter-layer links being one-way, as it removes the possibility of any rebound effects across action dimensions; specifically, even if we add a single link in the reverse direction, it could undermine the uniqueness of the NE. For instance, assume the only non-zero entry in $G^{(21)}$ is $g^{(21)}_{ii}=1$. Then, simplifying the condition in Proposition~\ref{prop:twolayer-multilayer-general} for this game, if $(2 - [(I+G^{(1)})^{-1}]_{ii}\times g^{(12)}_{ii})<0$ (which can happen for any sufficiently large rebound inter-layer link $g^{(12)}_{ii}$), the multilayer may not have a unique NE. 

\subsubsection{Games of complements.} Lastly, we consider a case of games of complements, where bounding the extent of inter-layer interactions can guarantee NE uniqueness. 
\vspace{0.05in}
\begin{proposition}\label{prop:multilayer-BDD}
In a multilayer network where the layers are games of strategic complements, $I+G^{(l)}\in \mathcal{P}$, \revr{with $||A||_2 = \sup_{x\neq 0} \frac{||Ax||_2}{||x||_2}$ denoting the spectral matrix norm,} if
    \begin{align*}
    &||(I+G^{(1)})^{-1}||_2\times ||G^{(12)}||_2 < 1, \text{ and, }\\
    &||(I+G^{(2)})^{-1}||_2\times ||G^{(21)}||_2 < 1,
\end{align*}
the multilayer network game has a unique Nash equilibrium. 
\end{proposition}
\vspace{0.05in}
% \vspace{0.1in}
\emph{Intuitive interpretation.} Proposition~\ref{prop:multilayer-BDD} tells us that, for a game of strategic complements, knowing that the within-layer links have relatively higher influence than the inter-layer links is sufficient to make the game have a unique Nash equilibrium.

For the special case of symmetric matrices, the conditions reduce to $\lambda_{\max}(G^{(lk)})<1+\lambda_{\min}(G^{(l)})$. % and $\lambda_{\max}(G^{(21)})<1+\lambda_{\min}(G^{(2)})$. 
In words, this means that if the inter-layer connectivities (as measured by $\lambda_{\max}(G^{(lk)})$) surpasses the ability of the intra-layer links to mute the ups and downs in that layer (as measured by $\lambda_{\min}(G^{(l)})$), then the multilayer game will not be guaranteed to have a unique equilibrium. 
This also highlights why it is harder to guarantee equilibrium uniqueness in multilayer network games compared to a multiplex network games, due to the additional rebounds from inter-layer spillovers. 

\section{Stability}
In this section, we show that the identified conditions for the uniqueness of the NE of (multiplex and multilayer) network games also imply the stability of that equilibrium. 

We first formally define stability. 
For an LCP, stability is defined as ``low'' sensitivity to small changes in the LCP$(M,\mathbf{b})$ \cite{CuDuong1985}. Denote the solution set of LCP$(M,\mathbf{b})$ by SOL$(M,\mathbf{b})$. The solution $\mathbf{x}^\star$ is said to be a \emph{stable} solution to the LCP if there are neighborhoods $V$ of $\mathbf{x}^\star$ and $U$ of $(M, \mathbf{b})$ such that:
\begin{itemize}
    \item For all $(\Bar{M}, \Bar{\mathbf{b}})\in U$ the set SOL$(\Bar{M}, \Bar{\mathbf{b}})\cap V$ is non empty.
    \item  $\{\sup{||\mathbf{y} - \mathbf{x}^\star||: \mathbf{y} \in \text{SOL}(\Bar{M}, \Bar{\mathbf{b}})\cap V}\} \to 0$ as $(\Bar{M}, \Bar{\mathbf{b}})$ approaches $(M,\mathbf{b})$.
\end{itemize}
If, in addition to the above conditions, the set SOL$(\Bar{M}, \Bar{\mathbf{b}})\cap V$ is a singleton, then the
solution $\mathbf{x}^\star$ is said to be \emph{strongly stable}. {In the context of network games, such perturbations may be indicative of (small) network dynamics (reflected in differences between $M$ and $\bar{M}$), or errors in modeling agents' utility functions (reflected in differences between $\mathbf{b}$ and $\bar{\mathbf{b}}$).}

The following proposition shows that the uniqueness of the NE of a single-layer network game also implies that that equilibrium will be strongly stable. 

\vspace{0.05in}
\begin{proposition}\label{prop:stability}
    If a network game has a unique Nash equilibrium, then this equilibrium is strongly stable. 
\end{proposition}
\vspace{0.05in}

% \vspace{0.1in}
\emph{Remark.} For the special case of interior Nash equilibria (i.e., those in which all agents exert positive effort), the reverse is also true: interior solutions are stable if and only if they are unique. Intuitively, this is because unstable equilibria are caused by free-riding agents (i.e., those exhibiting zero effort) who may become active agents under perturbations of the game parameters, causing instability and multiplicity in the new emerging outcomes depending on how the resulting spillovers of their activity alter other agents' efforts. 

Note also that this proposition only states that if a network game has a unique Nash equilibrium, then this equilibrium is also strongly stable. 
However, the converse is not true, in that a network game can have multiple strongly stable Nash equilibria, as we show through the following examples. 

\vspace{0.05in}
\begin{example}[Stability in multiplex network games]
    Consider a multiplex network game with $G^{(1)} = \begin{psmallmatrix}
    0 & 1/2a & 1/2a \\ 
    1/2a & 0 & 1/2b \\
    1/2a & 3/2b & 0 \end{psmallmatrix}$ and $G^{(2)} = \begin{psmallmatrix}
    0 & 1/2a & 1/2a \\ 
    1/2a & 0 & 1/2b \\
    1/2a & -1/2b & 0 \end{psmallmatrix}$, This will lead to $I+G^\parallel=\begin{psmallmatrix}
    1 & 1/a & 1/a \\ 
    1/a & 1 & 1/b \\
    1/a & 1/b & 1 \end{psmallmatrix}$. Let us choose $b_i(.)$ and $c_i$ such that $\mathbf{q}= (-8, 3, -3)^T$. In this game, if we set $(a, b)=(2, 1/2)$, the first layer has \emph{non-unique and unstable} NE, the second layer has a \emph{unique and stable} NE, and the multiplex \emph{has non-unique, but stable} NE $\mathbf{x^\star}=(8, 0, 0)^T$ and $\mathbf{x^\star}=(16/3, 0, 5/3)^T$. Notably, in this example, connecting the two layers leads to stable equilibria, despite one of the layers not having stable equilibria by itself.   
\end{example}

\vspace{0.05in}
\begin{example}[Stability in multilayer network games]
    Consider a multilayer network with supra-adjacency matrix:
    % \begin{align*}%\label{ex:nunique-stable1}
        $I+G^\merge = \begin{psmallmatrix}
            1 & 2 & 1 & 0 \\
            3 & 1 & 0 & 1 \\
            1 & 0 & 1 & 3/5 \\
            0 & 1 & 4/5 & 1
        \end{psmallmatrix}$
    % \end{align*}
    with the top left $2\times 2$ sub-matrix being $G^{(1)}$ and the bottom right $2\times 2$ sub-matrix being $G^{(2)}$. Choose $b_i(.)$ and $c_i$ such that $\mathbf{q}=(-3, -6, -1, -0.5)^T$. Then, layer 1, by itself, has \emph{non-unique and unstable} NE, while layer 2 has \emph{unique and stable} NE. The multilayer network on the other hand, has \emph{non-unique equilibria} $\mathbf{x^\star}=(0.6, 1.8, 0, 0)^T$ and $\mathbf{x^\star}=(0, 6, 1, 0)^T$, with the former being \emph{stable} and the latter being \emph{unstable}. This example shows that having one layer with non-unique and unstable Nash equilibria will not necessarily ruin the stability of the solutions of the resulting multilayer game (although it always undermines the uniqueness of the NE). 
\end{example}

\section{Numerical Experiments}\label{sec:simulations}

\begin{figure*}[ht]
    \centering
    \vspace{-0.32in}
    \begin{subfigure}[t]{0.24\textwidth}
        \includegraphics[width=\textwidth]{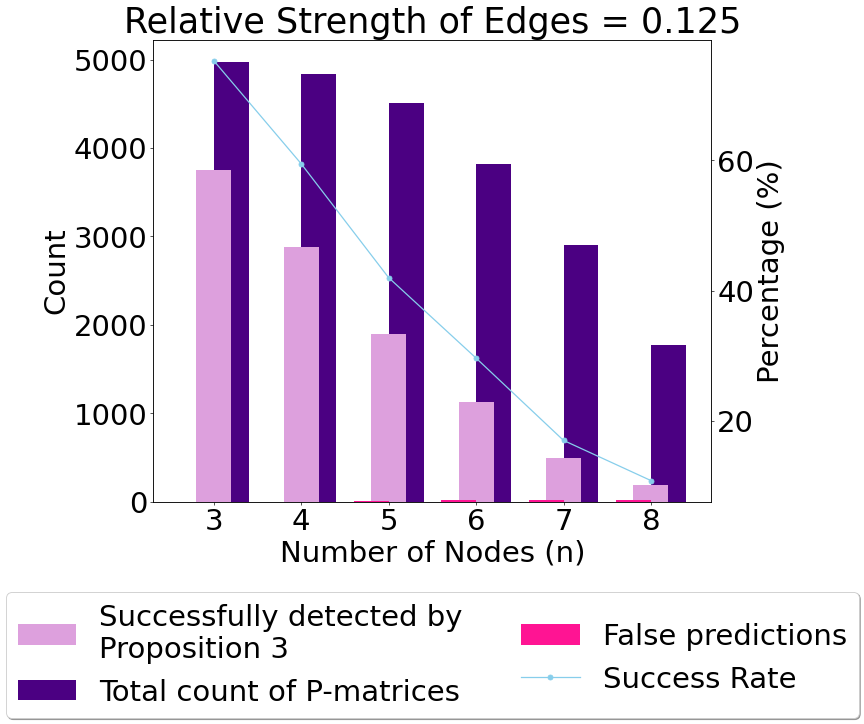}
            \caption{Proposition~\ref{prop:multiplex-pert-pos} for directed networks (for fixed \revr{strength ($s$)})}
        \label{fig:multiplex-prop-pos}
    \end{subfigure}
    % \hspace{0.01in}
    \begin{subfigure}[t]{0.24\textwidth}
        \includegraphics[width=\textwidth]{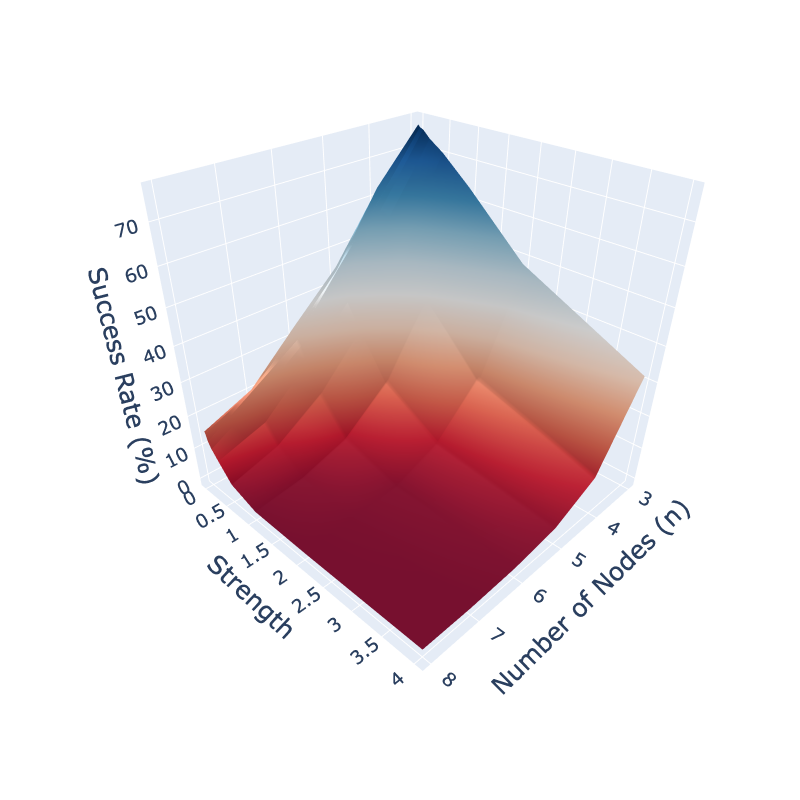}
        \caption{Proposition~\ref{prop:multiplex-pert-pos} for directed networks (when varying \revr{strength ($s$)})}
        \label{fig:prop-pos-var-s}
    \end{subfigure}
    \hspace{0.01in}
    \begin{subfigure}[t]{0.24\textwidth}
        \includegraphics[width=\textwidth]{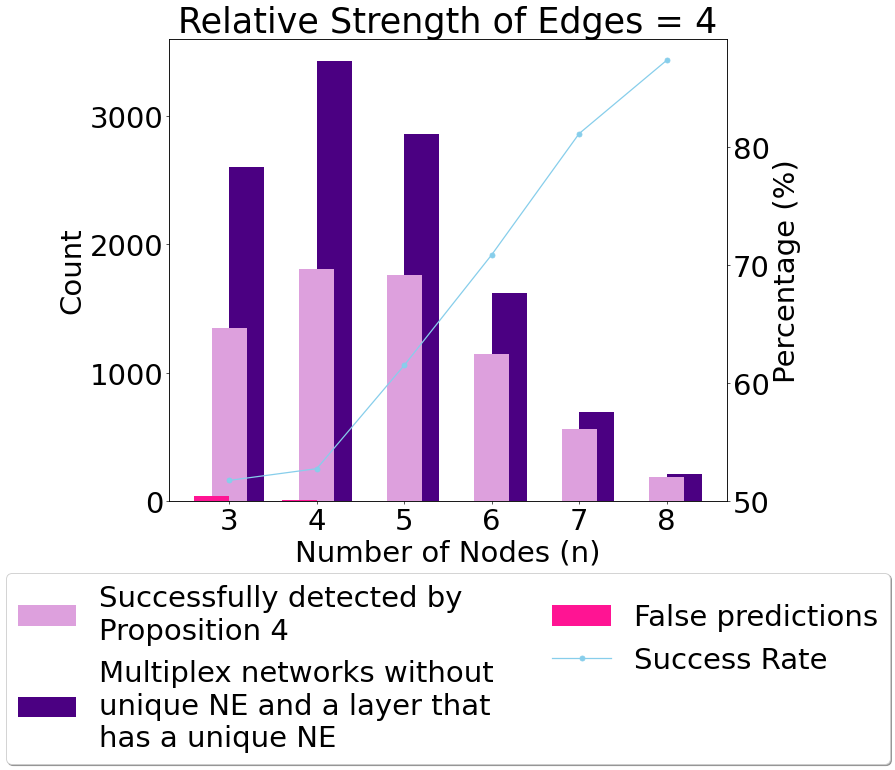}
            \caption{Proposition~\ref{prop:lambda-min-fails} for directed networks (for fixed \revr{strength ($s$)})}
        \label{fig:multiplex-prop-neg}
    \end{subfigure}
    % \hspace{0.01in}
    \begin{subfigure}[t]{0.24\textwidth}
        \includegraphics[width=\textwidth]{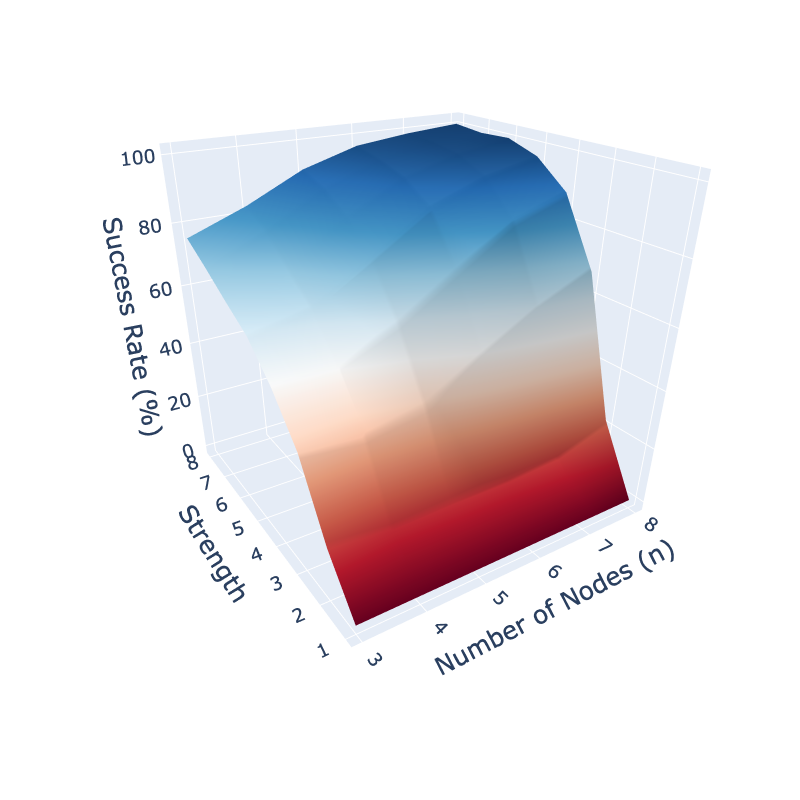}
            \caption{Proposition~\ref{prop:lambda-min-fails} for directed networks (when varying \revr{strength ($s$)})}
        \label{fig:prop-neg-var-s}
    \end{subfigure}
    \caption{\revr{By using the results in Proposition~\ref{prop:multiplex-pert-pos} and Proposition~\ref{prop:lambda-min-fails} in a directed setting we observe that }Proposition~\ref{prop:multiplex-pert-pos} is most effective at providing guarantees for the uniqueness of Nash equilibria in smaller multiplex networks (\ref{fig:multiplex-prop-pos} and \ref{fig:prop-pos-var-s}) and Proposition~\ref{prop:lambda-min-fails} is most effective at refuting the guarantees for the uniqueness of Nash equilibria in larger multiplex networks (\ref{fig:multiplex-prop-neg} and \ref{fig:prop-neg-var-s}). }
    %\vspace{-0.1in}
    \label{fig:prop4-plots-all}
    \vspace{-0.1in}
\end{figure*}

In this section, we present numerical experiments on randomly generated networks to verify our propositions, and also evaluate them when the assumptions in the propositions are not satisfied. We also provide two experiments based on real-world data (the Copenhagen Networks Study \cite{Sapiezynski2019}, and data Friendfeed and Twitter social network interactions \cite{Magnani2011MLmodel}, both multiplex networks) in \revr{the online appendix~\cite{online-appendix}}. 

\subsubsection{Multiplex Networks}
Recall that for multiplex network games with {undirected} underlying layers, Proposition \ref{prop:multiplex-pert-pos} identifies a condition under which the multiplex will have a unique NE, while Proposition \ref{prop:lambda-min-fails} identifies a condition under which the game is not guaranteed to have a unique NE. We now ask whether these conditions can still be informative when applied to \emph{{directed}} networks. {We note that in directed networks with asymmetric adjacency matrices, the eigenvalues can be complex numbers. When this happens, we use the real part of the eigenvalues to order the eigenvalues and assess the conditions of Propositions~\ref{prop:multiplex-pert-pos} and \ref{prop:lambda-min-fails}.}% ... which relates to P-matrix property, as seen in \cite{Fiedler1962OnMW}. 
% We avoid using the absolute values for comparison since the problems such as the following can occur:
% \begin{example}
%     Imagine we have a network with the adjacency matrix $ A = \begin{psmallmatrix}
%          0 & -3.33 & -0.49 & -0.76\\
%          2.17 & 0 & 2.64 & -3.87\\
%          -3.39 & -1.82 & 0 & -3.94\\
%          0.14 & 1.66 & 0.63 & 0
%     \end{psmallmatrix}$ with $\lambda_1 = 0.512 + i (0.540)$, $\lambda_2 = 0.512 - i (0.540)$, $\lambda_3 = -0.512 + i (4.437)$, and $\lambda_4 = -0.512 - i (4.437)$. As seen in the example, $\lambda_3$ and $\lambda_4$ will be chosen as $\lambda_{\min}$ while they have the largest absolute value, raising a mismatch between the choice of ``smallest'' and ``largerst'' and the comparison. 
% \end{example}

% In Proposition~\ref{prop:lambda-min-fails} we are already using absolute value of the negative eigenvalue ($\lambda_{\min}$), therefore, the use of absolute value in the complex plane will have the same notion of ``distance'' and will impact the inequality less. However, this is not the case for the Proposition~\ref{prop:multiplex-pert-pos}, by choosing the eigenvalue with smallest real part, we could even be choosing the eigenvalue with largest size, as in the example. These kinds of mismatches can be problematic in directed multiplex networks, causing the Proposition~\ref{prop:multiplex-pert-pos} to fail.} 

For our experiments, we consider networks of size (number of agents) $N\in[3,8]$. For each $N$, we generate 5000 instances of random (directed) networks by drawing the edge weight of the first layer from a uniform distribution with range $(-1, 1)$, and those of the second layer from a uniform distribution with range $(-s, s)$ where $s$ represents the strength of links in layer 2 relative to layer 1. We set $\kappa = 0.5$.  Of these instances, we only keep those in which layer 1 has a unique NE. {We then assess Propositions~\ref{prop:multiplex-pert-pos} or \ref{prop:lambda-min-fails} for each instance. We define an outcome to be a ``false detection'' when the inequalities in Proposition ~\ref{prop:multiplex-pert-pos} (resp. Proposition~\ref{prop:lambda-min-fails}) are satisfied (i.e., they predict that a NE is guaranteed (resp. not guaranteed) to be unique), but that the (directed) multiplex network does not follow this prediction.}

We first look at instances where the second layer has a weaker strength relative to the first layer (setting $s=0.125$). This way, it is more likely that the second layer is relatively sparse compared to the first layer, so that it can not undermine the uniqueness of the NE; intuitively, this is the condition that Proposition~\ref{prop:multiplex-pert-pos} assesses in order to guarantee uniqueness of the NE of an {undirected}  multiplex game. In Figure~\ref{fig:multiplex-prop-pos},  we observe that Proposition~\ref{prop:multiplex-pert-pos} can also often successfully identify if a \emph{{directed}} multiplex game has a unique NE in \emph{smaller} size networks. That is, this proposition successfully assesses when (a small number of) weak links from a newly added layer do not substantially alter the original layer's equilibrium state. 

We then look at instances where the second layer has stronger interactions relative to the first layer (setting $s=4$). This way, it is more likely that the ups-and-downs caused by a strong second layer can not be muted by the first layer, causing non-uniqueness of NE; intuitively, this is the condition identified by Proposition~\ref{prop:lambda-min-fails} for {undirected} multiplex games. In Figure~\ref{fig:multiplex-prop-neg}, we observe that Proposition~\ref{prop:lambda-min-fails} is also adept at identifying when \emph{{directed}} multiplex games lack a unique NE in \emph{larger} networks. That is, this proposition successfully assesses when (a large number of) stronger links from a new layer significantly impact the original layer's equilibrium state. 

{We also take a closer look at the failure cases of Propositions~\ref{prop:multiplex-pert-pos} and \ref{prop:lambda-min-fails}. To this end, we extracted two of the counter examples (false detection cases) from our numerical experiments (one for each of Proposition 3 and Proposition 4).

%In more detail, for Proposition 3, we have the following example of a failed detection: % when $N=5$ (no false detection appeared in the instances with $N=3$ and $N=4$). 
%For Proposition~\ref{prop:multiplex-pert-pos}, when generalized for asymmetric matrices, we could not construct any examples of false detection for the cases with $N=3$ and $N=4$. However, for $N=5$ the following example could be constructed:
\begin{example}\label{ex:counter-3}
    Consider a multiplex network with the layers' adjacency matrices are $G^{(1)} = \begin{psmallmatrix}
         0 & -0.20  & 0.20 & 0.39 & 0.16\\
         0.91 & 0 & -0.85 & -0.88 & 0.89\\
         -0.95 & -0.32 & 0 & -0.65 & 0.68\\
         -0.37 & -0.93 & -0.46 & 0 & 0.29\\
         0.57 & -0.02 & -0.73 & -0.15 & 0
    \end{psmallmatrix}$ and $G^{(2)} = \begin{psmallmatrix}
         0 & 0.11  & -0.08 & 0.05 & -0.09\\
         -0.08 & 0 & -0.07 & -0.06 & 0\\
         -0.05 & -0.10 & 0 & -0.10 & 0.08\\
         0.06 & -0.09 & 0.12 & 0 & 0.07\\
         -0.02 & -0.05 & -0.09 & 0.07 & 0
    \end{psmallmatrix}$, and $\kappa=0.75$. For these,  $\lambda_{\min}(G^{(1)})=-0.62$ and $\lambda_{\max}(G^{(2)})=0.05+0.02j$. In this case the inequality in Proposition~\ref{prop:multiplex-pert-pos} is satisfied;  however, $I+G^{\parallel}\notin \mathcal{P}$, which means that we cannot guarantee the uniqueness of the multiplex network game.
\end{example}

%For Proposition 4, on the other hand, we consider the following example: % of a failed detection when $N=3$ (no false detection appeared in the instances with $N\geq 7$). 
\begin{example}\label{ex:counter-4}
    Consider a multiplex network with $G^{(1)} = \begin{psmallmatrix}
         0 & 0.92  & -0.78\\
         -0.96 & 0 & -0.93\\
         0.95 & 0.79 & 0
    \end{psmallmatrix}$, $G^{(2)} = \begin{psmallmatrix}
         0 & -0.09  & -0.08\\
         -2.40 & 0 & -3.24\\
         -1.85 & -1.71 & 0
    \end{psmallmatrix}$, and $\kappa=0.5$. For these, we have $\lambda_{\max}(G^{(1)})=0.05+1.54j$ and $\lambda_{\min}(G^{(2)})=-2.51$. The inequality in Proposition 4 (generalized) will be satisfied since $2.51 \ge 2(1+0.05)$, claiming that the uniqueness of NE is not guaranteed.  However, the constructed multiplex from these two matrices is a P-matrix, and therefore the equilibrium uniqueness is, in fact, guaranteed. 
\end{example}}
{Intuitively, in Example~\ref{ex:counter-4}, the connectivity of layer 1 is in fact high enough to mute the ups-and-downs introduced by layer 2, yet the real part of the largest eigenvalue of layer 1 does not successfully capture this high connectivity. A similar discussion can be made for Example~\ref{ex:counter-3}. These suggest that extensions of  Propositions~\ref{prop:multiplex-pert-pos} and \ref{prop:lambda-min-fails} to directed networks should go beyond the real part of complex eigenvalues.}

{Finally, we conduct similar experiments for different values of $s$: 
%(relative strength of links in the second layer compared to the first layer): for Proposition~\ref{prop:multiplex-pert-pos}, we consider $0 < s\le 4$, and for Proposition~\ref{prop:lambda-min-fails}, we considered $1 \le s \le 8$. Our experiment setting for generating random instances is otherwise the same as described earlier. 
%As seen in Figure~\ref{fig:multiplex-prop-pos}, Proposition 3 generalizes well when the number of nodes and strength is small and Proposition 4 generalizes well when strength is large.} 
%Finally, we look into the general regions where Proposition~\ref{prop:multiplex-pert-pos} and Proposition~\ref{prop:lambda-min-fails} are able to be informative if the networks are directed. While Proposition~\ref{prop:lambda-min-fails} is more reliable if the networks are directed, and Proposition~\ref{prop:multiplex-pert-pos} is easier to fail. 
As seen in Figure~\ref{fig:prop-pos-var-s}, Proposition~\ref{prop:multiplex-pert-pos} generalizes more reliably when both the number of nodes and the relative strength of links in the second layer are small. This is explained by, and supports, the intuition on Proposition~\ref{prop:multiplex-pert-pos}: this proposition is most informative when the effect of layer 2 can be viewed as a (small) disturbance in layer 1, so that the multilayer network mostly resembles layer 1. 

We also see from Figure~\ref{fig:prop-neg-var-s} that a generalization of Proposition~\ref{prop:lambda-min-fails} to directed networks is more successful for large $s$ and large $N$; this is consistent with the intuition of Proposition~\ref{prop:lambda-min-fails} that the first layer may not be able to mute the ups-and-downs caused by a strongly connected second layer (in the sense of both having more links and stronger connections).}
%We define false detections as whenever the inequalities in Proposition~\ref{prop:multiplex-pert-pos} and Proposition~\ref{prop:lambda-min-fails} are satisfied, i.e. they guarantee and do not guarantee uniqueness respectively, but the multiplex network does not follow their prediction, i.e. the multiplex does not have or has unique Nash equilibrium. These cases rarely occur as seen in Figure~\ref{fig:multiplex-prop-neg} and Figure~\ref{fig:multiplex-prop-pos}, however they can be of interest as they can be costly mistakes. 

\subsubsection{Multilayer Networks} We next run experiments on random instances of multilayer network games as follows. We consider networks of size $N=5$.  We first chose 6 first layer networks with a unique Nash equilibrium, and edge weights drawn from $U(-1,1)$. We then generated 1000 second layer networks with edge weights drawn randomly from $U(-1,1)$, and paired them with each of the 6 first layers, for a total of 6000 instances. In Figure~\ref{fig:nume-exp-multilayer}, we sort these 6000 instances first by the minimum eigenvalue of the first layer, and then by the minimum eigenvalue of the second layer. Pink (lighter) colors indicate that the instance is guaranteed to have a unique NE. 

For each instance, we further consider 4 different strengths for the between-layer links, by drawing random weights from the following ranges: $(-1,1)$ for ``normal'' and ``one-way'' inter-layer links, $(-0.5,0.5)$ for ``weak'' inter-layer links, and $(-0.05,0.05)$ for ``very weak'' inter-layer links.

\begin{figure}[h]
    \centering
    \includegraphics[width=0.8\columnwidth]{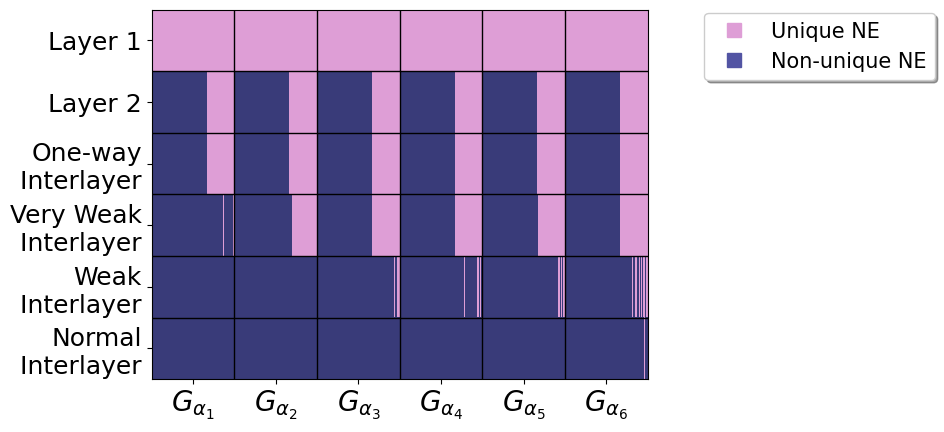}
    \caption{Uniqueness of equilibria of multilayer network games (in pink) is increasingly undermined as the strength of the interlayer interactions increases \revr{(left to right)}.}
    \label{fig:nume-exp-multilayer}
    \vspace{-0.1in}
\end{figure}

We first note the similarities between ``one-way" and ``very weak'' inter-layer links. Proposition~\ref{prop:multilayer-oneway} noted that one-way interactions can guarantee that joining two layers with unique NE will result in a multilayer game with unique NE; Figure~\ref{fig:nume-exp-multilayer} suggests that this is true as long as the inter-layer interactions are sufficiently weak, as well. However, as the strength of inter-layer interactions grow, we see fewer multilayer networks with a unique NE, as illustrated by the decrease in pink (lighter color) lines as we move down in Figure~\ref{fig:nume-exp-multilayer}. That said, we also note that as the minimum eigenvalue of the layers grow (moving to the right in Figure~\ref{fig:nume-exp-multilayer}), we are seeing more resistance from the layers against having the uniqueness of NE ruined by inter-layer interactions. This can be seen as the growing ability of the layers in damping the fluctuations introduced due to the rebounds between the action dimensions.%
\section{Conclusion}\label{sec:conclusion}
We have proposed two new classes of multiplex and multilayer network games to study networked strategic interactions when agents are affected by different modalities of information and operate over multiple action dimensions, respectively. These models enabled us to explore how the properties of the constituent subnetworks undermine or support the uniqueness and stability of the Nash equilibria of multiplex/multilayer games. At a technical level, answering these questions for multiplex (resp. multilayer) games required us to understand how the determinant (and lowest eigenvalues) of the sum (resp. block) adjacency matrix of the game relates to the determinants (and eigenvalues) of the layers' matrices; neither the determinant nor the eigenvalues of the sum/block matrices have such closed-form expressions in general. Our results have leveraged existing inequalities/bounds (e.g., on matrix perturbations, Weyl’s inequality) to find (sufficient) negative results, and provided positive answers for special matrix subclasses. 

Our findings shed light on the reasons for the fragility of the uniqueness and stability of equilibria when agents interact over networks of networks, and can guide potential interventions. For instance, we noted that the connectivity of one layer (as characterized by its largest eigenvalue) needs to be high enough to mute the rebounds introduced by another layer (as characterized by its lowest eigenvalue) as a necessary condition for equilibrium uniqueness; this suggests interventions in which a policy maker or network designer attempts to change either the connectivity or the bipartiteness of one of two interdependent networks to induce unique or stable equilibria. 

{We conclude with some potential directions of future work. First, our models and findings in this paper consider \emph{static} networks. Our results may be applicable to dynamic networks for which spectral properties can be appropriately bounded, with the provided bounds fitting our (inequality) conditions for equilibrium uniqueness (e.g., \cite{Chen2007} shows that the P-matrix property will continue to hold under sufficiently small perturbations to the adjacency matrix). Formally establishing such results, as well as establishing stronger uniqueness conditions for uniqueness (or lack thereof) for dynamic multilayer/multiplex networks, remain as directions of future work. Other} future directions of research include sharpening our stability results, combining the multiplex and multilayer network models, {providing uniqueness guarantees that leverage network structural properties beyond spectral properties of the adjacency matrix}, and analyzing multiplex/multilayer network games with non-linear best-replies.

\bibliographystyle{ieeetr}
\bibliography{multilayer-ne}

\clearpage
\appendix

\subsection{{Proofs}}\label{app:proofs}

\subsubsection{Proof of Proposition~\ref{prop:pair-failure}}
Consider any pair of agents $i$ and $j$. The principal minor corresponding to this pair in the multiplex network is $M_{ij}:=\begin{psmallmatrix}
    1 & \kappa g^{(1)}_{ij} + (1-\kappa) g^{(2)}_{ij} \\ \kappa g^{(1)}_{ji} + (1-\kappa) g^{(2)}_{ji} & 1
\end{psmallmatrix}$. A necessary condition for $I+G^\parallel$ to be a P-matrix is for this principal minor to have a positive determinant. Writing this condition in terms of the determinants of the pair of agents' corresponding principal minors $M^\alpha_{ij}$ and $M^\beta_{ij}$ in each layer, setting it to be non-positive, yields the proposition's statement. 

\subsubsection{Proof of Proposition~\ref{prop:sum-of-p-matrices}}
\emph{Part (1).} This is true because a symmetric matrix is a P-matrix if and only if it is positive definite \cite{murty1988linear}, and the sum of positive definite matrices is positive definite. 
\emph{Part (2).} By the triangle inequality, $|\kappa g^{(1)}_{ij} + (1-\kappa) g^{(2)}_{ij}|<\kappa |g^{(1)}_{ij}| + (1-\kappa)|g^{(2)}_{ij}|$. Therefore, we have $\sum_{j\neq i}|\kappa g^{(1)}_{ij} + (1-\kappa) g^{(2)}_{ij}|<1, ~\forall i$, and as such, $I+G^\parallel$ is strictly row diagonally dominant as well. by the Gershgorin circle theorem, all real eigenvalues of a strictly row diagonally dominant matrix with positive diagonal elements are positive. The determinant of a matrix is the product of its eigenvalues, and as for real matrices, the complex eigenvalues appear in pairs with their conjugates, $I+G^\parallel$ has a positive determinant. The same argument holds for all principal minors of $I+G^\parallel$. Therefore, $I+G^\parallel\in\mathcal{P}$.  
\emph{Part (3).} A B-matrix is a subclass of P-matrices~\cite{pena2001class}. It is easy to check that given that $I+G^{(1)}$ and $I+G^{(2)}$ are B-matrices, $I+G^\parallel$ also satisfies the conditions of a B-matrix, and is therefore a P-matrix.

\subsubsection{Proof of Proposition~\ref{prop:multiplex-pert-pos}}
From \cite{Chen2007}, we know that if $M$ is a P-matrix, every matrix $A\in\mathcal{M}\coloneqq\{A \,|\, \beta_p(M)\lVert M-A\rVert_p < 1\}$ is also a P-matrix, where $\beta_p(M)$ is:
\begin{align}
    \beta_p(M)=\max_{d\in[0,1]}\lVert(I-D+MD)^{-1}D)\rVert_p
\end{align}
and $\lVert . \rVert_p$ is the matrix norm induced by the vector norm for $p\ge 1$, and $D$ is diag($d_i$) such that $0\le d_i\le 1$ for all $i$. Intuitively, the matrices $A\in\mathcal{M}$ can be viewed as sufficiently small perturbations of matrix $M$. 
By setting $M=\kappa(I+G^{(1)})$ and $A=I+G^\parallel$, we see that for a multiplex network game where the first layer has a unique NE (i.e., $I+G^{(1)}\in\mathcal{P}$), the multiplex will have a unique NE (i.e., $I+G^\parallel \in\mathcal{P}$), if:
\begin{align}\label{eq:perturbation-chen}
    \beta_p(\,\kappa(I+G^{(1)})\, )\times(1-\kappa)\lVert (I+G^{(2)})\rVert_p < 1.
\end{align}

For a symmetric positive definite matrix $M$, we know $\beta_2(M) = \lVert M^{-1}\rVert_2$, so we can simplify \eqref{eq:perturbation-chen}:
\begin{align}\label{eq:pert-pos}
    \frac{1}{\kappa}\lVert (I+G^{(1)})^{-1}\rVert_2 \times (1-\kappa)\lambda_{\max}(I+G^{(2)})<1~.
\end{align}
    
Now, as $I+G^{(1)}$ does not have a negative eigenvalue, we can write $\lVert(I+G^{(1)})^{-1}\rVert_2 = \frac{1}{\lambda_{\min}(I+G^{(1)})}$. Therefore, \eqref{eq:pert-pos} will reduce to the condition in the proposition statement. %, completing the proof. 

\subsubsection{Proof of Proposition~\ref{prop:lambda-min-fails}}

Since by Corollary~\ref{cor:single-uniqueness-symmetric} we only require a bound on the minimum eigenvalue of $G^\parallel$, we consider \emph{Weyl's inequalities} \cite{tao-blog}, which provide an ordering of the eigenvalues of the sum of two symmetric matrices. Formally, let $H=H_1+H_2$, $H_1$, and $H_2$ be $n\times n$ Hermitian matrices, with their respective eigenvalues $\lambda_i$ indexed in decreasing order, i.e., $\lambda_{\max}=\lambda_1\geq \lambda_2\geq \ldots \geq \lambda_n=\lambda_{\min}$.  
Then, the following inequalities hold:
\begin{align*}%\label{weylineq}
    &\lambda_j(H_1)+\lambda_k(H_2)\le \lambda_i(H) \le \lambda_r(H_1) + \lambda_s(H_2)\notag\\
    &\text{s.t.} \hspace{0.2in} j+k-n\ge i \ge r+s-1~.
\end{align*}
 
At $i=n$, for $G^\parallel=\kappa G^{(1)} + (1-\kappa) G^{(2)}$, we have: %Weyl's inequalities lead to
\begin{align}
    \lambda_{\min} (G^\parallel) \leq \min_{\substack{r, s, \text{ s.t. } \\ r+s-1\leq n}} (\kappa\lambda_r(G^{(1)}) + (1-\kappa)\lambda_s(G^{(2)}))~.
    \label{eq:weylineq-min}
\end{align}

We now note that $tr(G^\parallel)=0$ ({as $g^{(1)}_{ii}=g^{(2)}_{ii}=0$}), and therefore $\lambda_{\min}(G^\parallel)<0$. As a result, $|\lambda_{\min}(G^\parallel)|\leq 1$ is the same as identifying conditions under which $\lambda_{\min}(G^\parallel)\leq -1$. Consider the term in the upper bound of \eqref{eq:weylineq-min} attained at $\{r=n,s=1\}$: % and $\{r=1,s=n\}$:
\begin{align*}
    \lambda_{\min}(G^\parallel)\le \kappa\lambda_{\max}(G^{(1)})+(1-\kappa)\lambda_{\min}(G^{(2)})~.
\end{align*}
If the upper bound above is less than $-1$, then $\lambda_{\min}(G^\parallel)<-1$, and the multiplex will not be guaranteed to have a unique NE. Re-arranging the inequality, and noting that $\lambda_{\min}(G^{(2)})<0$ and $\lambda_{\max}(G^{(1)})>0$ (as the traces for both of these matrices, and therefore the sum of their eigenvalues, is equal to zero), leads to the statement of the proposition. 

\subsubsection{Proof of Proposition~\ref{prop:twolayer-multilayer-general}}

For the first case, if we want $I+G^\merge$ to be a P-matrix, we need both $I+G_{1}$ and $I+G_{2}$ to be P-matrices as well, since these are sub-matrices of $I+G^\merge$. 

For the second case, note that we need the determinant of $I+G^\merge$ to be positive for it to be a P-matrix as a necessary condition. This can be written as: \begin{align}\label{eq:blockdeterminant}
&\det(I+G^\merge)=\det(I+G^{(1)})\times\notag\\
    &\qquad \det(I+G^{(2)}-G^{(21)}(I+G^{(1)})^{-1}G^{(12)})
\end{align}
where the second term is the Schur complement of $G^\merge$ \cite{boyd2004convex} (note that the condition for non-singularity to apply this result is satisfied, as $I+G^{(1)}$ is a P-matrix). 

For \eqref{eq:blockdeterminant} to be positive (and noting that $\det(I+G^{(1)})>0$ since we have assumed this layer has a unique NE), we need the second term on the RHS, the determinant of the Schur complement, to be positive. The same can be written for $I+G^{(2)}$, leading to the conditions in the second case.

\subsubsection{Proof of Proposition~\ref{prop:multilayer-oneway}}
For any chosen subset $\gamma$ of nodes, we can divide it into two disjoint subsets, such that $\gamma=\gamma_1\cup\gamma_2$, and %such that $\gamma_1 \cap \gamma_2 = \varnothing$, and 
with $\gamma_1$ containing the indices in $\{1, \ldots, N\}$, and $\gamma_2$ containing the indices in $\{N+1, ..., 2N\}$. This way, we can write any chosen sub-matrix of $I+G^\merge$ as follows:
    \begin{align}
        (I+G^\merge)_{[\gamma_1;\gamma_2]} = \begin{psmallmatrix}
          (I+G^{(1)})_{[\gamma_1]} & G^{(12)}_{[\gamma_1;\gamma_2]} \\
            0 & (I+G^{(2)})_{[\gamma_2]}
        \end{psmallmatrix}
    \end{align}
where for a matrix $A$, $A_{[\gamma]}$ is the square sub-matrix of $A$ with the rows and columns indexed in the set $\gamma$, and $A_{[\gamma_1;\gamma_2]}$ is the square sub-matrix of $A$ where row indices are in set $\gamma_1$ and column indices are in set $\gamma_2$. 

These sub-matrices' determinants are given by $\det((I+G^{(1)})_{[\gamma_1]})\times \det((I+G^{(2)})_{[\gamma_2]})$, which will be positive provided the layers are such that $I+G^{(l)}\in \mathcal{P}$.

\subsubsection{Proof of Proposition~\ref{prop:multilayer-BDD}}
We begin with a definition \cite{Feingold1962BlockDD}: 
a block matrix $G$, with nonsingular diagonal sub-matrices $G^{(i)}$, is strictly block diagonally dominant (sBDD) (with respect to norm $||.||_2$) if 
\begin{align}\label{eq:BDD}
(||G^{(l)^{-1}}||_2)^{-1} <  ||G^{(lk)}||_2, \forall \{l, k\neq l\} \in\{1,2\}. 
\end{align}
From \cite{Feingold1962BlockDD}, we also know that sBDD matrices with positive real diagonal entries have positive eigenvalues. Now, under the conditions in the proposition, $I+G^\merge$ is an sBDD matrix, and it also has positive real diagonals. Therefore, under the conditions of the proposition, $I+G^\merge$ is a positive definite matrix, and the resulting network game has a unique NE.

\subsubsection{Proof of Proposition~\ref{prop:stability}}
For an LCP$(M,\mathbf{b})$, with solution $\mathbf{z}^*=M\mathbf{x}^*+\mathbf{b}$, define:
\begin{align*}
    J \coloneqq \{i\; |\; x_i^*>0 , z_i^*=0\} \\%\;(r)\\
    K \coloneqq \{i\; |\; x_i^*=0 , z_i^*=0\} \\%\;(s)\\
    L \coloneqq \{i\; |\; x_i^*=0 , z_i^*>0\}
\end{align*}
Accordingly, define $M_r := \begin{psmallmatrix}
    M_{JJ} & M_{JK} \\
    M_{KJ} & M_{KK}
\end{psmallmatrix}$. This block matrix is a rearrangement of our original matrix $M$: 
%\begin{align}\label{eq:matrixperm}
$
    M_r = \prod_{i,j\in J,K}P_{ij} M P_{ij}, 
    $
%\end{align}
where $P_{ij}$ are permutation matrices.

From \cite{CuDuong1985}, we know the LCP$(M, \mathbf{b})$ is strongly stable at $\mathbf{x}^*$ if and only if the following conditions hold: (1) $M_{JJ}$ is nonsingular, and (2) The Schur complement $N\coloneqq M_{KK}-M_{KJ}M_{JJ}^{-1}M_{JK}$ is a P-matrix. 

First, note that if a single-layer network game has a unique Nash equilibrium, i.e., if $I+G$ is a P-matrix, then the rearrangement of $I+G$ as $I+G_r = \begin{psmallmatrix}
    I+G_{JJ} & G_{JK} \\
    G_{KJ} & I+G_{KK}
\end{psmallmatrix}$ will also be a P-matrix (this is because for every permutation matrix $Q$ and P-matrix $A$, the matrix $QAQ^T$ is also a P-matrix, which is a direct consequence of determinantal properties). Now, first note that since $I+G_r$ is a P-matrix, and $I+G_{JJ}$ is a square sub-matrix of $I+G_r$, we have $\det(I+G_{JJ})>0$, so that $I+G_{JJ}$ is non-singular. We further know that the Schur complement of a P-matrix is also a P-matrix \cite{horn2013matrix}. This completes the proof. 

\subsection{Extension to Networks with $M>2$ Layers}\label{appendix:m-layers}
We can extend Proposition~\ref{prop:sum-of-p-matrices} on the uniqueness of NE of multiplex networks as follows:
\begin{corollary}\label{prop:sum-of-p-matrices-m-layers}
    Under any of the following, a multiplex network game on $M$ layers will have a unique Nash equilibrium: 
    \begin{enumerate}
        \item $G^{(l)}$ are \emph{symmetric} {P-matrices} for all $l\in\{1, \ldots M\}$.  
        \item $I+G^{(l)}$ are \emph{strictly row diagonally dominant}, i.e.,  $\sum_{j\neq i} |g^{(l)}_{ij}| < 1, \forall l\in\{1, \ldots M\}$. 
        \item $I+G^{(l)}$ are \emph{B-matrices}, i.e.,  $1 + \sum_{j} g^{(l)}_{ij} > 0$ and $\frac{1}{N}(1+\sum_{j\neq k} g^{(l)}_{ij}) > g^{(l)}_{ik}$, $\forall i, \forall k\neq i, l\in\{1, \ldots, M\}$. 
    \end{enumerate}
\end{corollary}
Similar to Proposition~\ref{prop:sum-of-p-matrices}, these statements note that the limits of spillovers on each layer will remain true when these multiple layers connect. For instance, in statement 2, as long as the externalities received from other agents are limited in all layers, they will also be limited when these layers connect. The same can be said for other statements. 

For multilayer networks, we build on Proposition~\ref{prop:twolayer-multilayer-general}: if we want to retain uniqueness guarantees after adding a layer $G^{(M+1)}$ to an existing multilayer network $G^\merge_{old}$, we have to ensure that $I+G^\merge_{new} = \begin{pmatrix}\label{layeraddition}
      I+G^{(M+1)} & \mathbf{G^{(M+1, i)}}\\
      \mathbf{G^{(i, M+1)}} & I+G^\merge_{old}
    \end{pmatrix}$ is a P-matrix. 
The following identifies conditions under which this may not happen. 
\begin{corollary}
    A multilayer network game created by adding a layer to an existing multilayer network game will not have a unique NE if:
    \begin{enumerate}
        \item $I+G^{(M+1)}\notin P$ (meaning our additional layer needs to have a unique NE separately); or
        \item the corresponding Schur complement, $I+G^\merge_{old} - \mathbf{G^{(i, M+1)}}(I+G^{(M+1)})^{-1}\mathbf{G^{(M+1, i)}}$, does not have a positive determinant. 
    \end{enumerate}
\end{corollary}
Again, similar to Proposition~\ref{prop:twolayer-multilayer-general}, we will not have a unique NE necessarily after adding a layer with an isolated unique NE. %The added layer must also be compatible with the rest of the network.

Furthermore, as we show below, for chain and cycle multilayer networks (illustrated in Figure \ref{fig:cycle-chain}), we can identify additional conditions where the uniqueness of the NE is undermined. 

\begin{figure}
    \centering
    \includegraphics[width=0.32\textwidth]{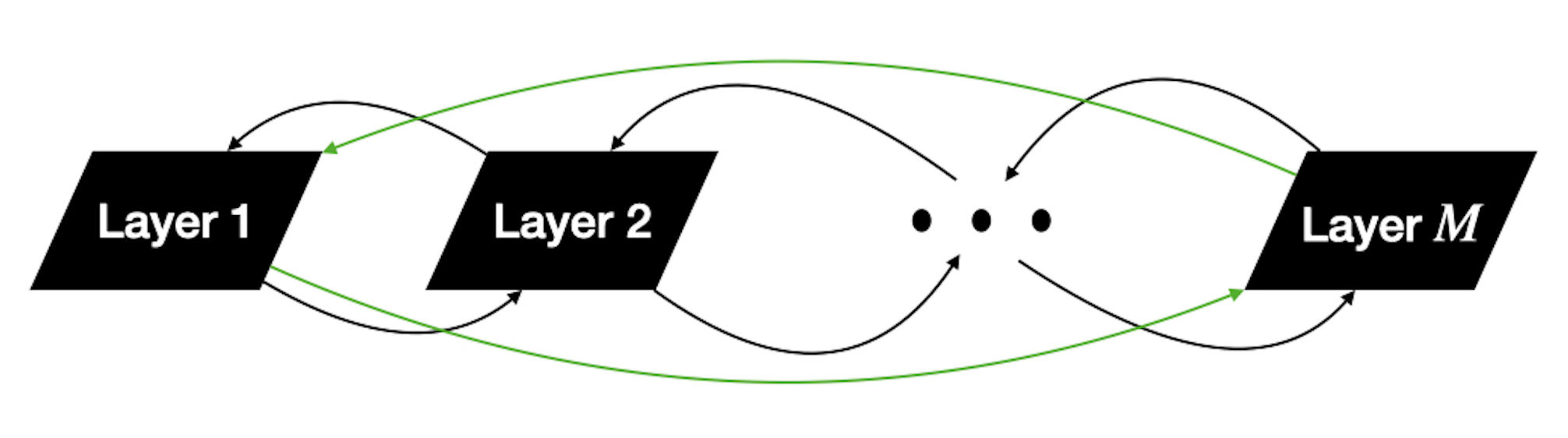}
    \caption{With (without) the green links, we will have a cycle (chain) of networks.}
    \label{fig:cycle-chain}
\end{figure}

\begin{proposition}
    Let $\mathcal{B}=\{B_1, B_2, ...\}$ denote the set of all non-diagonal blocks (matrices) in the supra-adjacency matrix. If any of the following cases exist in any of the matrices in $\mathcal{B}$, we will not have a unique NE:
    \begin{enumerate}
        \item A directed acyclic graph (DAG), i.e., a  directed graph with no directed cycles;
        \item An Idempotent matrix, i.e., a matrix which, when multiplied by itself, yields itself ($A^2=A$).
        \item A Circulant matrix with zero diagonal, i.e., a matrix where all row vectors are composed of the same elements, and each row vector is rotated one element to the right relative to the preceding row vector.
    \end{enumerate}
\end{proposition}
\begin{proof}
    For a \emph{block tridiagonal matrix} (cycle or chain), from \cite{MOLINARI20082221}, we know the determinant can be written as follows:
    \begin{align}
        &\det(I+G^\merge)=(-1)^{(M-1)n}\det(T-I_{2n})\det(B_1...B_M)\\
        &\det(I+G^\merge)=(-1)^{Mn}\det(T_0^{11})\det(B_1...B_M)
    \end{align}
Any of the statements above will make the last term equal to zero, and therefore, they will make the whole determinant zero, meaning that we will not have a unique NE anymore. 

(1): We will have an upper or lower triangular matrix with zero on the diagonal, which will give us a zero determinant.

(2): It is well known that the only non-singular Idempotent matrix is $I$.

(3) For a circulant matrix we can write its determinant as:
\begin{align}\label{eq:detcirculant}
    \det(C) = \prod_{j=0}^{n-1}(c_0+c_{n-1}\omega^j+...+c_1\omega^{(n-1)j})
\end{align}
Since in our case $c_0 = 0$, the first product in \eqref{eq:detcirculant} ($j=0$) will be zero, and therefore we will have a zero determinant. 
\end{proof}
\subsection{Existence of Nash Equilibria}\label{appendix:existence}

\subsubsection*{Existence on single-layer games}
The work of \cite{naghizadeh2017provision} identifies conditions for the existence of Nash equilibria in two particular classes of single-layer network games: games of strategic substitutes and games of strategic complements. Formally, in a game of strategic substitutes (complements), we have $g_{ij}\ge 0$ ($g_{ij}\le 0$), $\forall i, j\neq i$.
\begin{theorem}\cite[Theorem 2]{naghizadeh2017provision}\label{thm:strategicsubs}
    A single-layer network game of strategic substitutes (i.e., $g_{ij}\ge 0$, $\forall i, j\neq i$) always has at least one Nash equilibrium. 
\end{theorem}
\begin{theorem}\cite[Theorem 3]{naghizadeh2017provision}\label{thm:strategiccomp}
    A single-layer network game of strategic complements (i.e., $g_{ij}\le 0$, $\forall i, j\neq i$) has a (unique) Nash equilibrium if and only if $\rho(G) < 1$.
\end{theorem}

\subsubsection*{Existence on multiplex networks}

Building on the existing Theorems \ref{thm:strategicsubs} and \ref{thm:strategiccomp}, we explore conditions for the existence of Nash equilibria in multiplex network games of strategic substitutes and complements. 

\paragraph{Multiplex games of strategic substitutes} For multiplex network games of strategic substitutes (i.e., $g^{\parallel}_{ij}\geq 0, \forall i, j$), by Theorem \ref{thm:strategicsubs}, at least one Nash equilibrium always exists. These games can emerge if both layers are games of strategic substitutes (i.e., $g_{ij}^{(l)}\ge 0, \forall i,j\neq i$, and for $l\in\{1,2\}$). They can also emerge when the first layer is game of strategic substitutes (i.e., $g^{(1)}_{ij}\ge 0, \forall i, j, l$), and $\kappa\times\min_{ij}(g^{(1)}_{ij})\ge (1-\kappa)\times(|g^{(2)}_{ij}|)\;,\;\forall i,j$. 

\paragraph{Multiplex games of strategic complements} We next consider multiplex games of strategic complements (i.e., $g^{\parallel}_{ij}\leq 0, \forall i, j$). These can emerge when both layers are games of strategic complements (i.e., $g^{(l)}_{ij}\leq 0, \forall i, j, l$), as well as under mixed substitute and complement layers for which all weighted sum of edge weights remain negative. In this case, by Theorem \ref{thm:strategiccomp}, a Nash equilibrium exists in the multiplex game if and only if $\rho(G^\parallel)<1$. For general matrices $G^{(l)}$, there is no full characterization of the spectral radius of their sum $G^\parallel$. However, we can identify (sufficient) conditions for this bound to be satisfied in the case of \emph{{undirected}} networks, by leveraging Weyl's inequalities. 

\begin{proposition}\label{prop:multiplex-comp-exist}
    In a multiplex network game of complements {with undirected layers}, if 
    \begin{align*}
        |\kappa\lambda_{\min}(G^{(1)})+(1-\kappa)\lambda_{\min}(G^{(2)})|<1
    \end{align*}
at least one Nash equilibrium exists. 
\end{proposition}
\begin{proof}
The spectral radius of a matrix $A$ is defined as $\rho(A) = \max_i|\lambda_i(A)|$. For symmetric matrices, all eigenvalues are real, and we have $\rho(A) = \max\{-\lambda_{\min}(A), \lambda_{\max}(A)\}$, and by using Perron-Frobenius Theorem we know for a game of strategic complements $\rho(G)<1$ if and only if $|\lambda_{\min}(G)|<1$. 

Therefore, for $G^\parallel = \kappa G^{(1)}+(1-\kappa)G^{(2)}$, using Weyl's inequalities, we know:
    \begin{align}
        |\lambda_{\min}(G^\parallel)|&\le |\kappa\lambda_{\min}(G^{(1)})+(1-\kappa)\lambda_{\min}(G^{(2)})|\label{weylmin}
    \end{align}
Note that since $G^\parallel$ has zero on its diagonal entries, its smallest eigenvalue $\lambda_{\min}$ is negative. This means that if \eqref{weylmin} holds, then $\rho(G^\parallel)<1$, and by Theorem~\ref{thm:strategiccomp}, a (unique) NE to the multiplex network game will exist. 
\end{proof}

\emph{Intuitive interpretation.}  We note that the lowest eigenvalue represents the ``two-sidedness'' of a network. This means that if one or both constituent layers are two-sided (their $\lambda_{\min}$ have large magnitude), the condition in Proposition \ref{prop:multiplex-comp-exist} will be hard to satisfy; i.e., it will be harder to guarantee that the multiplex will have an equilibrium.  

Additionally, we identify conditions under which an NE does not exist in multiplex games of complements. 
\begin{proposition}\label{prop:multiplex-comp-not-exist}
    If any of the following conditions holds in a multiplex network game of complements {with undirected layers}, the game will not have a Nash equilibrium:
\begin{enumerate}
    \item $|\kappa\lambda_{\max}(G^{(1)})+(1-\kappa)\lambda_{\min}(G^{(2)})|\geq 1$,
    \item $|\kappa\lambda_{\min}(G^{(1)})+(1-\kappa)\lambda_{\max}(G^{(2)})|\geq 1$.
\end{enumerate}
\end{proposition}
The proof is straightforward and is based on finding lower bounds on the largest eigenvalue of $G^\parallel$ and lower bounds on the absolute value of its smallest eigenvalue using Weyl's inequalities. 

\emph{Intuitive interpretation.} For a case where we have two layers of strategic complements with a unique Nash equilibrium, i.e., $\rho(G^{(l)})=|\lambda_{\min}(G^{(l)})|<1$, we can say $|\kappa\lambda_{\min}(G^{(1)})+(1-\kappa)\lambda_{\min}(G^{(2)})| < 1$ which means neither of the conditions in Proposition~\ref{prop:multiplex-comp-not-exist} hold, since $|\lambda_{\min}(G^{(l)})|\ge\lambda_{\max}(G^{(l)})$ for a strategic game of complements.

Generally we can also say if $\rho(G^{(l)})>1$ for either layer, then there is an increased likelihood of not having a Nash equilibrium as per Proposition~\ref{prop:multiplex-comp-not-exist}, and it is also more challenging to satisfy the conditions for the existence of a Nash equilibrium as per Proposition~\ref{prop:multiplex-comp-exist}.

\subsubsection*{Existence on multilayer networks}

\paragraph{Multilayer games of strategic substitutes} We know from Theorem~\ref{thm:strategicsubs} that a  multilayer network game of strategic substitutes (i.e., $g_{ij}^\merge\geq 0, \forall i, j$) will have at least one Nash equilibrium. Such network can be constructed, e.g., if both layers are also games of strategic substitutes, \emph{and} the inter-layer connections are also non-negative, meaning $g_{ij}^{(pq)}\ge 0$.

\paragraph{Multilayer games of strategic complements} For a multilayer game of complements, we know from Theorem~\ref{thm:strategiccomp} that a Nash equilibrium exists if and only if $\rho(G^\merge) < 1$. The following lemma identifies a conditions when this happens. 

\begin{lemma}
    Consider a multiplex network game of complements with one-way inter-layer connection (i.e., with links directed only from one layer to the other, so that either $G^{(12)}=\textbf{0}$ or $G^{(21)}=\textbf{0}$). This game will have a (unique) Nash equilibrium if and only if $\rho(G^{(l)})<1, l\in\{1, 2\}$.
\end{lemma}
\begin{proof}
     In this scenario, the supra-adjacency matrix $G^\merge$ is a \emph{block triangular matrix}. Consequently, the eigenvalues of $I+G^\merge$, will be a union of the eigenvalues from individual layers $I+G^{(l)}$ for $l\in\{1,2\}$. Therefore, the multilayer network game's spectral radius ($\rho(G^\merge)$) will be less than 1 if and only if all constituent layers have a spectral radius ($\rho(G^{(l)})$) less than 1. According to Theorem~\ref{thm:strategiccomp}, this ensures the existence of a unique Nash equilibrium.
\end{proof}

\subsection{Experiments Based on Real-World Data}\label{app:real-world-experiments}

We conduct further experiments based on data from {two studies: (i) the Copenhagen Networks Study \cite{Sapiezynski2019}, and (ii) Friendfeed and Twitter social network interactions \cite{Magnani2011MLmodel}}. The former is a multiplex network, which connects a population of more than 700 university students over a period of four weeks. The layers are a ``calls'' network, an ``SMS'' network, and a ``Facebook friendship'' network. The first two layers are weighted directed networks, and the last one is a binary (unweighted and undirected) network. {The latter is a multiplex network connecting 2015 users on two social networks: Friendfeed and Twitter. Both layers are binary and undirected.}

\subsubsection{Copenhagen Study Data}
The data consists of 924 calls between 536 nodes for the ``calls'' network, and 1303 texts between 568 nodes in the ``SMS network''. We define the weights in the ``calls'' network as the duration of the call from $i$ to $j$ over the maximum duration of calls, and the weights in the ``SMS'' network as the number of times $i$ has sent a text to $j$ over the maximum number of texts in the network. We normalize both weights to be between 0 and 1. We then merged these ``SMS'' and ``calls'' layers to create a single \emph{in-person} network layer, $G_{in-person}$, as we hypothesize that the interactions in these networks are, most probably, reflective of overlapping in-person interactions. We also put higher weight on the ``calls'' network weights in this merging (as calls may reflect more influential connections than texts).

We then select {a random starting node and add the connected nodes to it until a community of $N=30$ nodes is reached (the algorithm resets at a new starting node to add if the resulting cluster is smaller than 30 nodes). We did this for 600 different cases, with 200 communities chosen from each of the calls, SMS, and Facebook friendship networks.} We assess when $I+G^\parallel_{total} = I + \kappa_1 G_{in-person} + (1-\kappa_1) G_{fb}$, the adjacency matrix of the multiplex network, has a positive determinant (this is a necessary condition to guarantee that the multiplex will have a unique, and hence stable, equilibrium), as {$\kappa_1$} (the influence of the in-person connections relative to the online connections) increases. We start with a powerful ``Facebook friendship'' layer ($\kappa_1$ close to $\frac{1}{5}$) and gradually move to a more powerful ``in-person layer'' ($\kappa_1$ close to 1). 

From Figure~\ref{fig:determinant_30}, we observe that negative determinants (and hence lack of guarantees for uniqueness and stability of the equilibria) can emerge when the ``Facebook friendship'' layer is more powerful than the ``in-person'' layer. On the other hand, as the influence of the ``in-person'' layer increases, we see that it has the potential to make up for the fluctuations and non-uniqueness caused by the ``Facebook friendship'' layer. 

\begin{figure}[h]
    \centering
    \vspace{-0.1in}
    \begin{subfigure}[c]{0.48\columnwidth}
        \includegraphics[width=\textwidth]{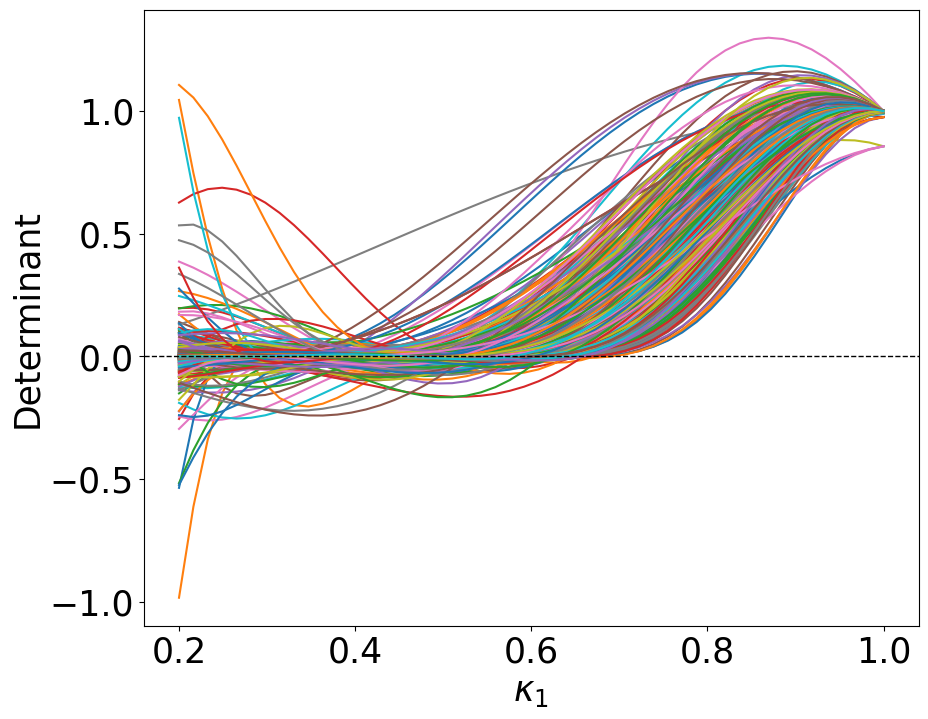}
        \caption{}
        \label{fig:determinant_30}
    \end{subfigure}
    % \hspace{0.01in}
    \begin{subfigure}[c]{0.48\columnwidth}
        \includegraphics[width=\textwidth]{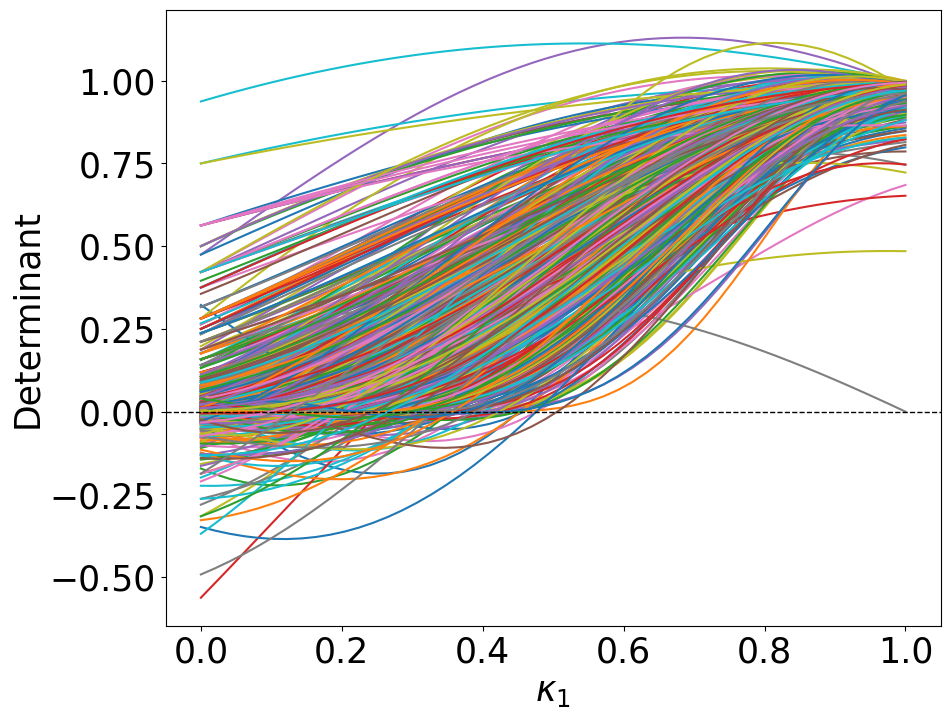}
        \caption{}
        \label{fig:determinant_NewData}
    \end{subfigure}
    \vspace{-0.1in}
    \caption{Determinant of 600 multiplex sub-networks, extracted from the Copenhagen Networks Data with $N=30$ (left) and from the Friendfeed-Twitter Data with $N=50$ (right), as $\kappa_1$ (the influence of in-person interactions) increases.}
    \vspace{-0.1in}
    \label{fig:prop4-plots}
\end{figure}

\subsubsection{Friendfeed and Twitter Data}
{This data consists of 2015 nodes with accounts in both social networks, with 17213 edges in the Friendfeed layer and 20191 edges in Twitter layer. We set the edge weights to reflect the users ``attention'' to each connection. For the Friendfeed layer, we assume each node uniformly divides their attention to all their outgoing links, i.e. $attention = \frac{1}{degree}$ (this choice will make the layer have a diagonally dominant adjacency matrix, and hence, a unique NE). For Twitter, we assume agents pay equal attention to all their connections, i.e. $attention = c$ where $c$ is a constant (so that this layer will not have a unique NE). We chose sub-networks similar to the previous experiment, but this time with $N=50$ for 600 different cases, with 300 of the communities were chosen from each of the Friendfeed and Twitter layers. 

As seen in Figure~\ref{fig:determinant_NewData}, a similar pattern to Figure~\ref{fig:determinant_30} emerges as the influence of the layer with a unique NE increases, implying that by having a multiplex with two layers where one has a unique NE while the other does not, one can construct a multiplex with a unique (and stable) NE by ensuring that the layer that has a unique NE is favored by users (has higher $\kappa$).}

\end{document}